\documentclass[11pt]{article}
\usepackage{xcolor}

\usepackage{graphicx,wrapfig,lipsum}
\usepackage{float}    
\usepackage{verbatim} 
\usepackage{amsmath, amssymb, amsthm,mathtools}  
\usepackage{subfig}   
\usepackage[colorlinks=true,citecolor=blue,linkcolor=blue,urlcolor=blue]{hyperref}
\usepackage{bookmark}
\usepackage{fullpage}
\usepackage{enumerate}
\usepackage{paralist}
\usepackage{xspace}
\usepackage{caption}
\usepackage{bbm}
\usepackage{bm}
\usepackage{url}
\usepackage{mathrsfs}
\usepackage{algorithm}
\usepackage{algorithmic}
\usepackage{color}
\usepackage{microtype}
\usepackage[numbers]{natbib}
\usepackage[section]{placeins}
\usepackage{lscape}
\usepackage{multirow}

\DeclareMathOperator{\supp}{supp}

\theoremstyle{definition}

\theoremstyle{plain}
\newtheorem{lem}{Lemma}

\newtheorem{prop}{Proposition}
\newtheorem{thm}{Theorem}
\newtheorem{assumption}{Assumption}

\newcommand{\abs}[1]{\left\lvert#1\right\rvert}
\newcommand{\norm}[1]{\left\lVert#1\right\rVert}

\newcommand{\lonenorm}[1]{\lVert#1\rVert_1}
\newcommand{\ltwonorm}[1]{\lVert#1\rVert_2}
\newcommand{\linfinity}[1]{\lVert#1\rVert_{\infty}}

\def \RR {\mathbb{R}}

\begin{document}

\title{Efficient Inference of Spatially-varying Gaussian Markov Random Fields with Applications in Gene Regulatory Networks\thanks{This research is supported, in part, by NSF Award DMS-2152776, ONR Award N00014-22-1-2127, NIH Award R37, MICDE Catalyst Grant, MIDAS PODS grant and “Startup funding from the University of Michigan”.
\textit{The first two authors contributed equally.}}}
\author{
  Visweswaran Ravikumar\thanks{Computational Medicine and Bioinformatics, University of Michigan} \and
  Tong Xu\thanks{Quantitative Finance and Risk Management, University of Michigan} \and
  Wajd N. Al-Holou \thanks{ Neurosurgery, University of Michigan}
  \and 
  Salar Fattahi\thanks{Industrial and Operations Engineering, University of Michigan, (co-corresponding author, \texttt{fattahi@umich.edu})} \and 
  Arvind Rao\thanks{Computational Medicine and Bioinformatics, University of Michigan, (co-corresponding author, \texttt{ukarvind@med.umich.edu})}
}
	
	\maketitle

\begin{abstract}
In this paper, we study the problem of inferring spatially-varying Gaussian Markov random fields (SV-GMRF) where the goal is to learn a network of sparse, context-specific GMRFs representing network relationships between genes. An important application of SV-GMRFs is in inference of gene regulatory networks from spatially-resolved transcriptomics datasets. The current work on inference of SV-GMRFs are based on the regularized maximum likelihood estimation (MLE) and suffer from overwhelmingly high computational cost due to their highly nonlinear nature. To alleviate this challenge, we propose a simple and efficient optimization problem in lieu of MLE that comes equipped with strong statistical and computational guarantees. Our proposed optimization problem is extremely efficient in practice: we can solve instances of SV-GMRFs with more than 2 million variables in less than 2 minutes. We apply the developed framework to study how gene regulatory networks in Glioblastoma are spatially rewired within tissue, and identify prominent activity of the transcription factor HES4 and ribosomal proteins as characterizing the gene expression network in the tumor peri-vascular niche that is known to harbor treatment resistant stem cells.
\end{abstract}

\section{Introduction}

The advent of high throughput sequencing technologies has transformed our understanding of biological systems, and catalyzed the adoption of a systems-level approach to studying biological processes . Networks have emerged as the intuitive framework for reasoning about complex biological systems \cite{barabasi2004network, zhang2014network}. Nodes in the network represent individual components, and edges represent direct interactions between them. For example, gene regualtory networks (GRNs) represent the wiring diagram of the cell's information processing system, with network edges identifying regulatory interactions between different genes. It has become clear that complex diseases like cancer must be understood at the level of this interactome, rather than the classical reductionist approach of studying individual components \cite{cheng2012understanding, du2015cancer}. As another example, with billions of neurons and hundreds of thousands of voxels, the human brain is considered as one of the most complex physiological networks, whose structure remains as a long-standing mystery~\cite{rubinov2010complex, huang2010learning, liu2013time, narayan2015two, kim2015highly}. The accurate inference of the brain connectivity network will have a far-reaching impact on understanding different neurological disorders~\cite{du2018classification, bassett2009human, menon2011large}. 
According to the NIH's \textit{BRAIN Initiative}, the development of ``{faster}, {less expensive}, and {scalable}'' technologies is the cornerstone for anatomic reconstruction of neural circuits at realistic scales~\cite{bargmann2014brain}.

Spatially resolved transcriptomics have emerged as a transformative technology in the recent past with immense potential to bolster our understanding of biology ant  a tissue architecture level\cite{marx2021method, lein2017promise, rao2021exploring}. Depending on the technology used, we can measure gene expression profile at near single cell resolution at the transcriptome-scale in situ \cite{asp2020spatially, moses2022museum}. In studying complex processes such as tumor growth, viewing cancer as a case of evolution within the tissue  has provided the groundwork for building a comprehensive theoretical framework to understand tumor diversity \cite{merlo2006cancer}. Evolutionary trade-offs between proliferation and survival strategies amongst cancer cells are driven by spatial gradients in exposure to nutrients, oxygen, immune cells and environmental toxicity between the tumor core versus periphery \cite{carmona2013emergence, carmona2017metabolic, hausser2020tumour}. The need to optimize growth of the tumor through evolution of the hallmark traits \cite{hanahan2011hallmarks} leads individual cancer cells to adopt a continuum of transcriptional states, that maximize their performance given spatially-imposed metabolic and survival constraints \cite{barkley2021cancer, groves2021cancer, becker2021single}. There are therefore strong spatial trends in the gene expression profiles and the underlying regulatory networks even amongst differentiated cells of the same type in both homeostasis and diseased states. Being able to infer these dynamic regulatory networks would provide us with a new lens for understanding complex biological processes, and can lead to new hypotheses regarding molecular mechanisms that would inspire further experimental and theoretical investigations into the nature of regulatory interactions underlying disease states.

One popular approach to model these problems is based on \textit{spatially-varying Markov Random Fields} (SV-MRFs). SV-MRFs are associated with a network of undirected \textit{Markov graphs} {$\mathcal{G}_{k}(V_{k},E_{k})$, where $V_{k}$ and $E_{k}$ are the set of nodes and edges in the graph at location $k$}. The node set $V_{k}$ represents the random variables (e.g. genes) in the model, while the edge set $E_{k}$ captures the conditional dependency between these variables at location $k$. In the special case of Gaussian Markov Random Fields (GMRFs), the edge set of the Markov graphs can be fully characterized based on the inverse covariance matrix (also known as the precision matrix). In particular, if the entry $(i,j)$ of the precision matrix $\Theta_k$ is zero, then the variables $i$ and $j$ at location $k$ are independent conditioned on the remaining variables.

A widely-used method for the inference of SV-MRFs is based on the so-called regularized maximum-likelihood estimation (MLE). Intuitively, MLE seeks to find a graphical model based on which the observed data is most likely to occur. However, MLE-based methods suffer from major computational challenges that undermine their applicability in large-scale settings. For example, in the Gaussian setting, the MLE requires optimizing over the so-called log-determinant of the precision matrix, which are  known to be intractable in large scales \cite{boyd2004convex, hallac2017network, fattahi2021scalable}. This drawback is further compounded in the spatially-varying regime, where the precision matrix must be estimated at each spatial location, leading to a dramatic increase in the size of the problem. \vspace{2mm}

\subsection{Our Contributions}
To address the aforementioned challenges, we propose a \textit{simple} estimation methods for the inference of spatially-varying GMRFs. Unlike MLE-based methods, our proposed approach is based on a class of simple and computationally efficient optimization methods that come equipped with strong statistical guarantees and are implementable in realistic scales. Our contributions are summarized as follows:\vspace{2mm}

\noindent\textbf{Computational guarantee:} Our proposed method reduces to a series of decomposable convex quadratic optimization problems that can be solved efficiently using any off-the-shelf solvers. In addition, the decomposable nature of the proposed optimization problem makes it amenable to parallel and distributed implementation.\vspace{2mm}

\noindent \textbf{Statistical guarantees: } In addition to the desirable computational guarantees, we show the statistical consistency of our proposed method---both theoretically and in practice. In particular, we characterize the non-asymptotic consistency of our proposed method, proving that it accurately recovers the underlying graphical model, even in the high-dimensional settings where number of available samples is significantly smaller than the number of unknown parameters. Moreover, it can efficiently reveal the correct sparsity information in the parameters and their differences.\vspace{2mm}  

	\noindent\textbf{Application in inferring gene regulatory networks:} Glioblastoma (GBM) is an incurable malignancy of the brain, with a median survival time of only 12-18 months despite therapy with surgical resection, chemotherapy and radiation \cite{hanif2017glioblastoma}. Despite aggressive treatment, these tumors inevitably recur and this recurrence is likely due to significant heterogeneity, which has been highlighted by single cell sequencing studies  \cite{yabo2021cancer}. Heterogenous populations of treatment-resistant tumor cells with stem cell properties have been identified in GBM that have been shown to drive treatment recurrence. Furthermore, these resistant cells often reside within unique microenvironmental niches \cite{al2021subclonal, hambardzumyan2015glioblastoma, osuka2017overcoming, lathia2011deadly}. The consequence of spatial context in regulating the tumor cell state, stemness properties, and treatment resistance in these tumors is increasingly appreciated \cite{kumar2019intra, jung2021tumor}. It is thus imperative that we understand how the gene networks of GBM cells are rewired as a function of their spatial environment, to identify context-specific upstream regulators of heterogenous tumor cell states. We thus employ our developed statistical framework to study how gene regulatory networks are spatially rewired in GBM. 
	
	 We partition the tumor section into distinct micro-environmental niches and estimate networks involving genes showing significant spatial trends in their activity. We identify Transcription factors and hub genes that control tumor behavior in distinct local environments. We find that the perivascular tumor niche is characterized by high levels of activity of ribosomal genes and that HES4 is a prominant upstream regulator in this environment. Our findings have been previously reported to be particularly important aspects of the stemness features of Glioblastomas in recent literature \cite{bazzoni2019role, shirakawa2020ribosomal}. We feel that our ability to define context specific upstream regulators of tumor states is an important step in fighting tumor recurrence and developing targeted therapies for this disease. 
	
\subsection{Notations} The $i$-th element of a vector $v$ or $v_t$ is denoted as $v_i$ or $v_{t;i}$. For a matrix $M$, the notations $M_{i:}$ and $M_{:j}$ denote the $i$-th row and $j$-th column, respectively. Moreover, for an index set $S$ and a matrix $M$, the notations $M_{S:}$ and $M_{:S}$ refer to a submatrix of $M$ with rows and columns indexed by $M$, respectively. For a matrix $M$ or a vector $v$, the notations $\|M\|_{\ell_q}$ and $\|v\|_{q}$ correspond to the element-wise $\ell_q$-norm of $M$ and $\ell_q$-norm of $v$, respectively. Moreover, $\|M\|_q$ and $\|M\|_{\max}$ are the induced $q$-norm and the element with the largest absolute value of the matrix $M$, respectively. Moreover, $\|M\|_0$ denotes the total number of nonzero elements in $M$. We use $M\succ0$ to show that $M$ is positive definite. For a vector $v$ and matrix $M$, the notations $\operatorname{supp}(v)$ and $\operatorname{supp}(M)$ are defined as the index sets of their nonzero elements. Given two sequences $f(n)$ and $g(n)$ indexed by $n$, the notation $f(n) \lesssim g(n)$ implies $f(n) \leq C g(n)$ for some constant $C<\infty$. Moreover, the notation $f(n) \asymp g(n)$ implies that $f(n) \lesssim g(n)$ and $g(n) \lesssim f(n)$. The sign function $\mathrm{sign}(\cdot)$ is defined as $\mathrm{sign}(x) = x/|x|$ if $x\not=0$ and $\mathrm{sign}(0) = 0$. Accordingly, when $x$ is a vector, the function $\mathrm{sign}(x)$ is defined as $\begin{bmatrix}\mathrm{sign}(x_1) & \mathrm{sign}(x_2) & \dots & \mathrm{sign}(x_n)\end{bmatrix}^\top$.
\vspace{2mm}

\noindent {\bf Organization.} The rest of the paper is organized as follows. In Section~\ref{sec:problem}, we formulate the inference of spatially-varying GMRFs and discuss the shortcomings of the existing techniques. Motivated by these shortcomings, we present a new formulation of the problem in Section~\ref{sec:proposed}. The related work is presented in Section~\ref{sec_relatedwork}. In Section~\ref{sec:main}, we delineate the statistical guarantees of our proposed formulation, and how to solve it efficiently. Finally, we showcase the performance of our proposed method on synthetically generated as well as the Glioblastoma spatial transcriptomics dataset in Section~\ref{sec:numerical}.

\section{Problem Formulation}\label{sec:problem}

Consider data samples from $K$ different Gaussian distributions with covariance matrices $\Sigma_{k}^\star \in \mathbb{S}_{+}^d,  k=1,...,K$ and sparse precision matrices $\Theta^\star_{k} = {\Sigma^\star_k}^{-1},  k=1,...,K$. Let $\{x^{k}_i\}_{i=1}^{n_k}$ be $n_k$ independent samples drawn from the $k$-th distribution, i.e., $x_i^{k} \sim \mathcal{N}(0, \Sigma^\star_k)$, for every $i = 1,\dots, n_k$ and $k=1,\dots,K$. Our goal is to estimate the precision matrices $\{\Theta^\star_{k}\}_{k=1}^K$ given the samples. The most commonly-used method to perform this task is via \textit{maximum likelihood estimation} (MLE) with an $\ell_1$ regularizer (also known as \textit{Graphical Lasso}~\cite{friedman2008sparse}):
\begin{align*}
    \widehat{\Theta}_k = \arg\min_{\Theta_k\succ 0} \mathrm{Tr}(\Theta_k\widehat{\Sigma}_k)-\log\det(\Theta_k)+\lambda\|\Theta_k\|_{\ell_1}
\end{align*}
where $\mathrm{Tr}(\cdot)$ is the trace operator and $\widehat{\Sigma}_{k}\coloneqq \frac{1}{n_{k}} \sum_{i=1}^{n_{k}} x_{i}^{k} x_{i}^{k^{\top}}$ is the sample covariance matrix for distribution $k$. 
A major drawback of the above estimation method is that it ignores any common structure among different distributions. To address this issue, a common approach is to consider a joint estimation method (also known as \textit{joint Graphical Lasso}~\cite{danaher2012joint}):
\begin{align}\label{eq_joint_GL}
    \{\widehat{\Theta}_k\} \!=\! \arg\min_{\Theta_k\succ 0}& \sum_{k=1}^K\left(\!\mathrm{Tr}(\Theta_k\widehat{\Sigma}_k)\!-\!\log\det(\Theta_k)\!+\!\lambda\|\Theta_k\|_{\ell_1}\!\right)+ P\left(\{\Theta_k\}_{k=1}^K\right)
\end{align}
where the term $P\left(\{\Theta_k\}_{k=1}^K\right)$ is a penalty function that encourages similarity across different precision matrices.
A major difficulty in solving joint Graphical Lasso is its computational complexity: in order to obtain an $\epsilon$-accurate solution, typical numerical solvers for~\eqref{eq_joint_GL} have complexity ranging from $\mathcal{O}(Kd^6\log(1/\epsilon))$ (via general interior-point methods)~\cite{mohan2014node, potra2000interior} to $\mathcal{O}(Kd^3/\epsilon)$ (via tailored first-order methods, such as ADMM)~\cite{hallac2017network, danaher2012joint, Jing16}. Solvers with such computational complexity fall short of any practical use in the large-scale settings. Indeed, the prohibitive worst-case complexity of methods based on Graphical Lasso is also exemplified in their practical performance~\cite{fattahi2019graphical, zhang2018large, fattahi2021scalable, fattahi2018sparse, fattahi2019linear, fattahi2018closed}.

\section{Proposed Method}\label{sec:proposed}

To address the aforementioned issues, we propose the following surrogate optimization problem for estimating sparse precision matrices:

\begin{equation}\tag{Elem-$q$}
	\begin{aligned} 
	\{\widehat{\Theta}_k\} = \arg\min_{\Theta_k} \ &  \underbrace{\sum_{k=1}^K\|\Theta_k - \tilde{F}^*(\widehat{\Sigma}_{k})\|_{\ell_{2}}^2}_{\text{backward mapping deviation}} + \underbrace{\mu\sum_{k=1}^K \|{\Theta_k}\|_{\ell_1}}_{\text{absolute regularization}}+\underbrace{\gamma \sum_{l>k}W_{kl}\norm{\Theta_{k}-\Theta_{l}}_{\ell_q}^q}_{\text{spatial regularization}}\\
    \end{aligned}
	\label{problem:lq}
\end{equation}
In the above optimization, the \textit{backward mapping deviation} captures the distance between the estimated precision matrix and the so-called \textit{approximate backward mapping} which will be described in Section~\ref{subset_ABM}. Moreover, the \textit{absolute regularization} promotes sparsity in the estimated parameters, whereas \textit{spatial regularization} encourages common spatial similarities among different parameters. For any given pair $(k,l)$, the weight $W^{-1}_{kl}$ can be interpreted as the ``distance'' between the $k$-th and $l$-th MRFs. Accordingly, a large value for $W_{kl}$ encourages similarity between $\Theta_k$ and $\Theta_l$. Two common choices of spatial similarities are sparsity and smoothness:

\begin{itemize}
\item {\it Smoothly-changing GMRF:} In smoothly-changing GMRFs, the adjacent precision matrices vary gradually. In this setting, $q=2$ can be used as the spatial regularizer in~\ref{problem:lq} to promote the smoothness in the parameter differences.
    \item {\it Sparsely-changing GMRF:} In sparsely-changing GMRFs, the adjacent precision matrices differ only in a few entries. In this setting, $q=1$ is a natural choice for the spatial regularizer in~\ref{problem:lq} since it promotes sparsity in the parameter differences.
\end{itemize}

\subsection{Approximate Backward Mapping}\label{subset_ABM}

Our proposed optimization problem is contingent upon the availability of an approximate backward mapping. For a GMRF, the backward mapping is defined as the inverse of the true covariance matrix, i.e., $F^*(\Sigma^\star_k) = {\Sigma^\star_k}^{-1} = \Theta^\star_k$~\cite{wainwright2008graphical}. Based on this definition, a natural surrogate for the backward mapping is ${{F}}^*({{\widehat{\Sigma}}_k}) = {{\widehat{\Sigma}}_k}^{-1}$, where ${\widehat{\Sigma}}_k$ is the sample covariance matrix for distribution $k$. However, in the high-dimensional settings, the number of available samples is significantly smaller than the dimension, and as a result the sample covariance matrix ${\widehat{\Sigma}}_k$ is singular and non-invertible. To alleviate this issue,~\citet{Yang14} introduce an approximation of the backward mapping based on soft-thresholding. Consider the operator $\texttt{ST}_{\nu}(M):\mathbb{R}^{d\times d}\to \mathbb{R}^{d\times d}$, where $\texttt{ST}_{\nu}(M)_{ij} = M_{ij}-\mathrm{sign}(M_{ij})\min\{|M_{ij}|,\nu\}$ if $i\not=j$, and $\texttt{ST}_{\nu}(M)_{ij} = M_{ij}$ if $i=j$. Given this operator, the approximate backward mapping is defined as $\widetilde F^*(\widehat \Sigma_k) = \texttt{ST}_{\nu}(\widehat \Sigma_k)^{-1}$, for every $k=1,\dots,K$. An important property of this approximate backward mapping is that it is well-defined even in the high-dimensional setting $n_k\ll d$ with an appropriate choice of the threshold $\nu$~\cite{Yang14}. Given this approximate backward mapping, we will show that the estimated precision matrices from~\ref{problem:lq} are close to their true counterparts with an appropriate choice of parameters.

\subsection{Decomposability}

An important property of~\ref{problem:lq} is that it naturally decomposes over different coordinates of the precision matrices: for every $(i,j)$ with $1\leq i\leq j\leq d$, the $ij$-th element of $\{\Theta_{k}\}_{k=1}^K$ can be obtained by solving the following subproblem:
\begin{equation}\tag{Elem-$(i,j,q)$}
	\begin{aligned} 
	\{\widehat\Theta_{k;ij}\}_{k=1}^K =& \arg\min_{\{\Theta_{k;ij}\}_{k=1}^K}  \sum_{k=1}^K\left({\Theta_{k;ij} - [\tilde{F}^*(\widehat{\Sigma}_{k})]_{ij}}\right)^2\!+\!\mu\sum_{k=1}^K |{\Theta_{k;ij}}|+\gamma \!\!\sum_{l>k}\!\! W_{kl}\left|{\Theta_{k;ij}-\Theta_{l;ij}}\right|^q,
    \end{aligned}
	\label{problem:lq_ij}
\end{equation}
Recall that the original problem~\ref{problem:lq} has $Kd(d+1)/2$ variables. The above decomposition implies that~\ref{problem:lq} can be decomposed into $d(d+1)/2$ smaller subproblems, each with only $K$ variables that can be solved independently in parallel. This is in stark contrast with the joint Graphical Lasso, which requires a dense coupling among the elements of the precision matrices through the non-decomposable $\mathrm{logdet}$ function. Later, we will show how each subproblem can be solved efficiently for different choices of $q$.

\section{Related Work}\label{sec_relatedwork}

Recently, many approaches have been proposed for sparse precision matrix estimation in high dimensions. This line of work begins by the inference of a single precision matrix, which can be achieved by $\ell_1$-regularized MLE, also known as \emph{Graphical Lasso (GL)} \cite{friedman2007sparse, banerjee2007model, Yuan07}.

Extending beyond single precision matrix inference, a recent line of research has focused on estimating time-varying MRFs, where the relation among variables may change over time \cite{zhou2008time}. A common approach for estimating time-varying MRFs is based on kernel methods, where the sample covariance matrix at any given time is a weighted average of the samples over time, where the weights are collected from a predefined kernel~\cite{zhou2008time, greenewald2017timedependent, Fattahi21}. 

In the context of spatially-varying graphical models, the main focus has been devoted to different variants of MLE-based techniques, such as \textit{Fused Graphical Lasso} (FGL) and \textit{Group Graphical Lasso} (GGL)~\cite{danaher2012joint}. FGL penalizes the pairwise difference of the precision matrices in $\ell_1$-norm, while GGL regularizes the $\ell_{2}$-norm of the $(i, j)$-th element across all $K$ precision matrices. \citet{Guo11}  reparameterized each off-diagonal element as the product of a common factor and difference, then applied separate $\ell_{1}$ regularization to these two parts. \citet{Takumi16} proposed to regularize the MLE with a Laplacian-type penalty to exploit the information among different distributions. However, all these techniques are based on MLE, and consequently suffer from a notoriously high computational cost.

To alleviate the computational cost of MLE-based technique, \citet{lee15a} proposed to estimate the joint precision matrices based on a \textit{constrained $\ell_{1}$ minimization for inverse matrix estimation} (CLIME) technique \cite{cai2011constrained}. Unlike GL, CLIME does not optimize over the complex $\mathrm{logdet}$ function and has shown more favorable theoretical properties than GL. Finally, our method is built upon the \textit{Elementary Estimator} introduced by~\citet{Yang14}, where the proposed estimator admits a closed-form solution based on soft-thresholding. This method was later extended by~\citet{Fattahi21} to time-varying setting, showing that it can be solved in near-linear time and memory.

\section{Statistical Guarantees}\label{sec:main}
In this section, we elucidate the statistical properties of~\ref{problem:lq} for SV-GMRFs with two widely-used spatial structures, namely \textit{smoothly-changing} and \textit{sparsely-changing} GMRFs. 
To this goal, we first need to make two important assumptions on the true precision matrices.
\begin{assumption}[Bounded norm]\label{asp_bound}
	There exist constant numbers $\kappa_{1}<\infty, \kappa_{2}>0$, and $\kappa_{3}<\infty$ such that
	\begin{equation*}
		\left\|\Theta_{k}^{\star}\right\|_{\infty} \leq \kappa_{1}, \inf _{w:\|w\|_{\infty}=1}\left\|\Sigma_{k}^{\star} w\right\|_{\infty} \geq \kappa_{2},\left\|\Sigma_{k}^{\star}\right\|_{\max} \leq \kappa_{3}
		\label{Assump1}
	\end{equation*}
	for every $k=1, \ldots, K.$
\end{assumption}
Assumption~\ref{asp_bound} is fairly mild and implies that the true covariance matrices and their inverses have bounded norms.
\begin{assumption}[Weak sparsity]
	Each covariance matrix $\Sigma_{k}^{\star}$ satisfies $\max _{i} \sum_{j=1}^{d}\left|\left[\Sigma_{k}^{\star}\right]_{i j}\right|^{p} \leq s(p)$, for some function $s: [1,\infty) \rightarrow \mathbb{R}$ and scalar $0\leq p< 1$.
	\label{Assump2}
\end{assumption}
Informally, we say ``the true covariance matrices are weakly sparse'' if $\{\Sigma^\star_t\}_{t=0}^T$ are $s(p)$-weakly sparse with $s(p)\ll d$ for some $0\leq p< 1$. 
The notion of weak sparsity extends the classical notion of sparsity to dense matrices. 
Indeed, except for a few special cases, a sparse matrix does not have a sparse inverse. Consequently, a sparse precision matrix may {not} lead to a sparse covariance matrix. However, a large class of sparse precision matrices have weakly sparse inverses. For instance, if $\Theta^\star_{k}$ has a banded structure with small bandwidth, then it is known that the elements of $\Sigma^\star_{k}={\Theta^\star_{k}}^{-1}$ enjoy exponential decay away from the main diagonal elements \cite{Demko84, Kershaw70}. Under such circumstances, simple calculation implies that $s(p) \leq \frac{C}{1-\rho^{p}}$ for some constants $C>0$ and $\rho<1$. More generally, a similar statement holds for a class of inverse covariance matrices whose support graphs have large average path length \cite{Benzi07, Benzi15}; a large class of inverse covariance matrices with row- and column-sparse structures satisfy this condition. 

Next, we introduce some notations that simplify our subsequent analysis. Let $\pi:\{1,2,\dots,K\}^2\to \{1,2,\dots,K(K+1)/2\}$ be a fixed, predefined labeling function that assigns a label to each pair $(k,l)$ with $l\geq k$. Let $G$ be a diagonal matrix whose $k$-th diagonal entry is defined as $W^{1/q}_{\pi^{-1}(k)}$. Moreover, let $A\in\mathbb{R}^{K(K-1)/2\times K}$ be the adjacency matrix defined as $A(\pi(k,l), k) = 1$ and $A(\pi(k,l), l) = -1$, for every $l>k$. Finally, define $\Theta_{ij} = [\Theta_{1;ij}\ \Theta_{2;ij}\ \dots\ \Theta_{K;ij}]^\top$ and $\tilde{F}^*_{ij} = \left[[\tilde{F}^*(\widehat{\Sigma}_{1})]_{ij}\ [\tilde{F}^*(\widehat{\Sigma}_{2})]_{ij}\ \dots\ [\tilde{F}^*(\widehat{\Sigma}_{K})]_{ij}\right]$, for every $j\geq i$. It is easy to see that $\|GA\Theta_{ij}\|_q^q = \sum_{l>k}W_{kl}|\Theta_{k;ij}-\Theta_{l;ij}|^q$ for every $j\geq i$, and accordingly,~\ref{problem:lq_ij} can be written concisely as

\begin{equation}
	\begin{aligned} 
	\widehat\Theta_{ij} \!=\! \arg\min_{\Theta_{ij}}  \left\|{\Theta_{ij} \!-\! \tilde{F}^*_{ij}}\right\|^2_{2}\!+\!\mu \|{\Theta_{ij}}\|_{1}+\gamma \left\|GA\Theta_{ij}\right\|_{q}^q.
    \end{aligned}
	\label{problem:lq_ij2}
\end{equation}
Next, we provide sharp statistical guarantees for our proposed method when the precision matrices $\{\Theta^\star_k\}_{k=1}^K$ change smoothly or sparsely across different distributions.

\subsection{Smoothly-changing GMRF}

We start with our main assumption on the smoothness of the precision matrices.

\begin{assumption}[Smoothly-changing SV-GMRFs] There exists a constant $D\geq 0$ such that $\sum_{l\geq k}(\Theta^\star_{k;ij}-\Theta^\star_{l;ij})^2\leq D^2$ for every $(i,j)$.
\label{asmp_smooth}
\end{assumption}
Informally, we say ``SV-GMRF is smoothly-changing'' if Assumption~\ref{asmp_smooth} is satisfied with a small $D$. For a smoothly-changing SV-GMRF, it is natural to choose $q=2$ in~\ref{problem:lq} to promote smoothness in the spatial difference of the precision matrices. Our next theorem characterizes the sample complexity of~\ref{problem:lq} with $q=2$ for smoothly-changing SV-GMRF. Let $\Theta_{\min} = \min\{|\Theta^\star_{k;ij}|: \Theta^\star_{k;ij}\not=0\}$.

\begin{thm}[Smoothly-changing SV-GMRF]
	Consider a smoothly-changing SV-GMRF with parameter $D$, and weakly-sparse covariance matrices with parameter $s(p)$ for some $0\leq p<1$. Suppose that the number of samples satisfies 
	$$
	n_k \gtrsim L\frac{\log d}{\Theta^2_{\min}},\quad \text{ where }\quad  L \!=\! \max\left\{\!\left(\frac{s(p)}{\kappa_2}\right)^{\frac{2}{1-p}}\!\!\!\kappa_3^2,\! \left(\!\frac{\kappa_1\kappa_3}{\kappa_2}\!+\!D\!\right)^2\!\right\}.
	$$
	Define $n_{\min} = \min_{k}\{n_k\}$. Moreover, suppose that $\widetilde{F}^*(\widehat{\Sigma}_k) = [\texttt{ST}_{\nu_k}(\widehat{\Sigma}_{k})]^{-1}$ with $\nu_{k} \asymp \kappa_3\sqrt{{\log d}/{n_{k}}}$. Then, the solution obtained from~\ref{problem:lq} with $q=2$ and parameters  
	$$\gamma \asymp \frac{1}{K\norm{W}_{\max}}\sqrt{\frac{\log d}{{n_{\min}}}},\qquad \mu \asymp D\sqrt{\frac{\log d}{n_{\min}}},
	$$
	satisfies the following statements with probability of $1-Kd^{-10}$:
	\begin{itemize}
	    \item {\bf Sparsistency:} The solution is unique and satisfies $\supp(\widehat{\Theta}_k)=\supp({\Theta}^\star_k)$ for every $k$.
	    \item {\bf Estimation error:} The solution satisfies
	    \begin{align*}
        &\|{\widehat{\Theta}_{k} -\Theta_{k}^\star}\|_{\max} \lesssim \left(\frac{\kappa_1\kappa_3}{\kappa_2}+D\right)\sqrt{\frac{\log d}{n_{\min}}},\quad \text{for every $k$.}
    \end{align*}
	\end{itemize}
	\label{thm_smooth}
\end{thm}

For smoothly-changing SV-GMRF, the above theorem provides a non-asymptotic guarantee on the estimation error and sparsistency of the estimated precision matrices via~\ref{problem:lq} with $q=2$, proving that the required number of samples must scale only \textit{logarithmically} with the dimension $d$. Moreover, both the estimation error and the required number of samples decrease with a smaller smoothness parameter $D$; this is expected since a small value of $D$ implies that the adjacent distributions share more information, and hence, the SV-GMRF is easier to estimate.

\subsection{Sparsely-changing GMRF}
In sparsely-changing SV-GMRFs, the precision matrices are assumed to change sparsely across different distributions; this is formalized in our next assumption.

\begin{assumption}[Sparsely-changing SV-GMRFs] There exists a constant $D_0\geq 0$ such that $\sum_{l\geq k}\|(\Theta^\star_{k;ij}-\Theta^\star_{l;ij})\|_0\leq D_0$ for every $(i,j)$.
\label{asmp_sparse}
\end{assumption}

Similar to the smoothly-changing SV-GMRFs, we say ``SV-GMRFs is sparsely-changing'' if it satisfies Assumption~\ref{asmp_sparse} with a small $D_0$. 
For a sparsely-changing SV-GMRFs, it is natural to choose $q=1$ in~\ref{problem:lq} to promote sparsity in the spatial difference of the precision matrices. To analyze the statistical property of this problem, we first consider~\eqref{problem:lq_ij2} with $q=1$ and rewrite it as:
\begin{align}
	\min\ & \ltwonorm{\widetilde{F}^*_{ij} \!-\! \Theta_{ij}}^2 \!+\! \mu \lonenorm{B\Theta_{ij}}, \text{ where } B\!=\!\left[\begin{array}{c}
			\frac{\gamma}{\mu} GA \\
			I
		\end{array}\right].
	\label{lasso_problem_1}
\end{align}

The above reformulation is a special case of the \textit{generalized Lasso problem} introduced by~\citet{lee_2013}. To show the model selection consistency of the above formulation, we next introduce the notion of \textit{irrepresentability}.

For any fixed $(i,j)$, let $\mathcal{S}_B\subset \{1,2,...,K(K+1)/2\}$ be the support of $B\Theta_{ij}^\star$, i.e., $[B\Theta_{ij}^\star]_k \not = 0$ for every $k\in \mathcal{S}_B$. Moreover, let $\mathcal{S}^c_B = \{1,2,...,K(K+1)/2\}\backslash \mathcal{S}_B$. Evidently, we have $|\mathcal{S}_B| \leq D_0+S_0$, where $D_0$ is introduced in Assumption~\ref{asmp_sparse} and $S_0$ is defined as the maximum number of nonzero elements in $\Theta_{ij}^\star$, i.e., $S_0 = \max_{i,j}\{\|\Theta^\star_{ij}\|_{0}\}$.

\begin{assumption}[Irrepresentability condition (IC),~\citet{lee_2013}]\label{asp_IC}
We have
\begin{equation}
        \left\|B_{\mathcal{S}^c_B:}B_{\mathcal{S}_B:}^{\dagger} \operatorname{sign}\left((B\Theta_{ij}^{\star})_{\mathcal{S}_B:}\right)\right\|_{\infty} \leq 1-\alpha
        \label{eq:irrepresentability}
\end{equation}
for some $0 < \alpha \leq 1$, where $B_{\mathcal{S}_B:}^{\dagger}$ is the Moore-Penrose pseudo-inverse of a matrix $B_{\mathcal{S}_B:}$.
\end{assumption}
The irrepresentability condition (IC) entails that the rows of $B$ corresponding to the zero elements of $B\Theta^\star_{ij}$ must be nearly orthogonal to the other rows. Despite the seemingly complicated nature of IC, classical results on Lasso have shown that it is a necessary condition for the exact sparsity recovery, and hence, cannot be relaxed~\cite{zhao2006model, wainwright2009sharp}. Later, we show that this condition is satisfied for our problem under a mild condition on the weight matrix $W$ and parameters $\mu$ and $\gamma$.

Another quantity that plays a central role in our derived bounds is the so-called \textit{compatibility constant} defined as 
\begin{align*}
    \kappa_{\mathrm{IC}} := \left\|B_{\mathcal{S}^c_B:}B_{\mathcal{S}_B:}^{\dagger}\right\|_{\infty}+1.
\end{align*}
The compatibility constant $\kappa_{\mathrm{IC}}$ is closely related to IC. In particular, if $\left\|B_{\mathcal{S}^c_B:}B_{\mathcal{S}_B:}^{\dagger}\right\|_{\infty}\leq 1-\alpha$ (which is a slightly stronger version of IC), then $\kappa_{\mathrm{IC}}\leq 2-\alpha$. Similar to IC, we will later show that $\kappa_{\mathrm{IC}}$ remains bounded under a mild condition on the weight matrix $W$. Finally, we define $\Delta{\Theta}_{\min} = \min_{k,i,j}\{|\Theta^\star_{k;ij}-\Theta^\star_{l;ij}|: \Theta^\star_{k;ij} - \Theta^\star_{l;ij}\not=0\}$. 

\begin{thm}[Sparsely-changing SV-GMRFs.]\label{thm_sparse}
Consider a sparsely-changing SV-GMRFs with parameter $D_0$, and weakly-sparse covariance matrices with parameter $s(p)$ for some $0\leq p<1$. Suppose that the number of samples satisfies 
\begin{align*}
    n_{\min} &\gtrsim L \frac{\log d}{\min\{\Theta_{\min}^2,\Delta\Theta_{\min}^2\}},\quad
    \text{where}\quad  L &\!=\! \left\{\left(\frac{s(p)}{\kappa_2}\right)^{\frac{2}{1-p}}\!\!\!\kappa_3^2,\!\left(\frac{\kappa_{\mathrm{IC}}\kappa_1\kappa_3}{\kappa_2\alpha}\right)\left({\norm{W}_{\max}\!D_0\!+\!S_0}\right)\right\}.
\end{align*}
Define $n_{\min} = \min_{k}\{n_k\}$. Moreover, suppose that $\widetilde{F}^*(\widehat{\Sigma}_k) = [\texttt{ST}_{\nu_k}(\widehat{\Sigma}_{k})]^{-1}$ with $\nu_{k} \asymp \kappa_3\sqrt{{\log d}/{n_{k}}}$. Moreover, suppose that the weight matrix $W$ and parameters $\mu$ and $\gamma$ are chosen such that IC (Assumption~\ref{asp_IC}) is satisfied. Then, the solution obtained from~\ref{problem:lq} with $q=1$ and parameter  
	$$\mu\asymp\frac{\kappa_{\mathrm{IC}}\kappa_1\kappa_3}{\kappa_2\alpha}\sqrt{\frac{\log d}{n_{\min}}},
	$$ 
	satisfies the following statements with probability $1-Kd^{-10}$:
	\begin{itemize}
	    \item {\bf Sparsistency.} The solution is unique and satisfies $\supp(\widehat{\Theta}_k)=\supp({\Theta}^\star_k)$ for every $k$ and $\supp(\widehat{\Theta}_k-\widehat{\Theta}_l)=\supp({\Theta}^\star_k-{\Theta}^\star_l)$ for every $k>l$.
	    \item {\bf Estimation error.} For every $(i,j)$, the solution satisfies
	    \begin{align*}
	\left\|\widehat\Theta_{ij}\!-\!\Theta_{ij}^{\star}\right\|_{2}\!\!\lesssim\!\left(\!\sqrt{\norm{W}_{\max}\!D_0}\!+\!\sqrt{S_0}\!\right) \frac{\kappa_{\mathrm{IC}}\kappa_1\kappa_3}{\kappa_2\alpha}\sqrt{\frac{\log d}{n_{\min}}}.
\end{align*}
	\end{itemize}
\end{thm}
The above theorem characterizes the sample complexity of inferring sparsely-changing SV-GMRFs, showing that the sparsity pattern of the precision matrices and their differences can be recovered \textit{exactly}, given that the number of samples scale logarithmically with the dimension and the problem satisfies IC.
Evidently, our result crucially relies on the satisfaction of IC and $\kappa_{\mathrm{IC}}$ being small. This leads to a follow-up question: how restrictive are these conditions in practice? Our next proposition shows that both conditions hold if $\gamma$ and $\mu$ are selected such that $\mu\leq \gamma\leq 2\mu$ and $W_{kl}$ is the same for every $k>l$.

\begin{prop}\label{prop_IC}
    Suppose that $0<\mu\leq \gamma\leq 2\mu$ and $W_{kl}$ is the same for every $k>l$. Then, $1\leq \kappa_{\mathrm{IC}}\leq 5$ and IC holds with $\alpha = \mu/\gamma$.
\end{prop}

Proposition~\ref{prop_IC} can be easily extended to general choices of $W$. In particular, suppose that $W = \tau \mathbf{1}\mathbf{1}^\top+E$ for some $\tau>0$, where $\mathbf{1}$ is the vector of ones. Then, Proposition~\ref{prop_IC} combined with a simple matrix perturbation bound reveals that
$$
\alpha\geq 1/2-\mathcal{O}(\norm{E}_{\max}),\ \text{and}\ 1\leq \kappa_{\mathrm{IC}}\leq 5+\mathcal{O}(\norm{E}_{\max}).
$$
In other words, IC holds and $\kappa_{\mathrm{IC}}$ remains bounded, provided that $\norm{E}_{\max} = \mathcal{O}(1)$, that is, the elements of the weight matrix $W$ do not vary too much. Later in our numerical experiments, we will show that such choices of $W$ provide the best statistical results on both synthetically generated as well as gene expression datasets. 

\section{Parameter Tuning and Implementation}\label{sec_alg}

In this section, we explain different implementation aspects of our proposed method.
\vspace{1mm}

\noindent{\bf Tuning $\mathbf{W}$:} To obtain a solution for~\ref{problem:lq}, we first need to fine-tune the parameters $\mu,\gamma,\nu_k, W$ based on the available data samples. 
Recall that, for every pair $(k,l)$, the value of $W_{kl}^{-1}$ can be interpreted as the "distance" between precision matrices for distributions $k$ and $l$. Intuitively, $\Theta^\star_k$ and $\Theta^\star_l$ are close if their corresponding covariance matrices $\Sigma^\star_k$ and $\Sigma^\star_l$ are close. Therefore, to obtain an estimate of $W$, we first compute the distance between any pair of sample covariance matrices $D_{kl} = \norm{\widehat{\Sigma}_k-\widehat{\Sigma}_l}_{\ell_2}$, and then assign $W_{kl} = 1/(1+D_{kl})$ for every $k\not=l$.\vspace{1mm}

\noindent{\bf Tuning $\pmb{\mu, \gamma}$, and $\pmb{\nu_k}$:} Recall that the parameter $\mu$ controls the sparsity of the estimated precision matrices, whereas $\gamma$ penalizes their differences. Moreover, $\nu_k$ is the threshold for used in the proposed approximate backward mapping. In Theorems~\ref{thm_smooth} and~\ref{thm_sparse}, we provide an explicit value for these parameters that depend on the parameters of the true solution, which are not known \textit{a priori}. Without any prior knowledge on the true solution, these parameters can be selected by minimizing the \textit{extended Bayesian Information Criterion} (BIC) \cite{foygel2010extended}:
\begin{equation}\label{eq_bic}
\begin{aligned}
    (\widehat{\mu},\widehat{\gamma}, \widehat{\nu}) =& \arg\min_{\mu, \gamma, \nu}\ \mathrm{BIC}(\mu, \gamma, \nu),\quad \text{where}\\
    \mathrm{BIC}(\mu,\! \gamma,\! \nu) \!:=&\! \sum_{k=1}^{K}\!n_k[\mathrm{Tr}{(\widehat{\Sigma}_k\widehat{\Theta}_k(\mu,\! \gamma,\! \nu))} \!-\! \log \det \widehat{\Theta}_k(\mu,\! \gamma,\! \nu)] +\log(n_k) \mathrm{df}^{(k)} + 4 \mathrm{df}^{(k)}\log d ,
\end{aligned}
\end{equation}
In the above definition, $\widehat{\Theta}_k(\mu, \gamma, \nu)$ is the optimal solution of~\eqref{problem:lq} with parameters $(\mu, \gamma, \nu)$. Moreover, $\mathrm{df}^{(k)}$ is defined as the number of nonzero elements in $\widehat{\Theta}_{k}(\mu, \gamma, \nu)$. Theorems~\ref{thm_sparse} and~\ref{thm_smooth} suggest that $\gamma = C_1\sqrt{{\log d}/{n_{\min}}}$, $\mu=C_2\sqrt{{\log d}/{n_{\min}}}$, and $\nu_k = C_3\sqrt{{\log d}/{n_k}}$, where $C_1, C_2$, and $C_3$ are constants that depend on the true solution.
Therefore, to pick the parameters $\mu$, $\gamma$, and $\nu_k$, we perform a grid search over the constants $C_1$, $C_2$, and $C_3$, and pick those that minimize $\mathrm{BIC}(\mu,\gamma, \nu)$.
\vspace{1mm}

\noindent{\bf Algorithm:} Next, we explain a general algorithm for solving~\ref{problem:lq}. As mentioned before,~\ref{problem:lq} decomposes over different coordinates $(i,j)$, where each subproblem can be written as~\ref{problem:lq_ij}. This decomposition leads to a parallelizable algorithm, where each thread solves~\ref{problem:lq_ij}, for a subset of the coordinates $(i,j)$. This approach is oulined in Algorithm~\ref{algorithm}.

\begin{algorithm}
		\caption{General algorithm for solving~\ref{problem:lq}}
		\label{algorithm}
		\begin{algorithmic}
			\STATE {\bfseries Input:} Data samples $\{x_i^k\}$, parameters $(\mu, \gamma,\nu_k)$, and weight matrix $W$.;
			\STATE {\bfseries Output:} Solution $\{\widehat{\Theta}_k\}_{k=1}^K$ for~\ref{problem:lq};
			\STATE Compute the sample covariance matrix $\widehat{\Sigma}_{k}$ for every $k=1,...,K$;
			\STATE Compute the approximate backward mapping $\widetilde{F}^{*}(\widehat{\Sigma}_{k}) = \left[\texttt{ST}_{\nu_k}(\widehat{\Sigma}_{k})\right]^{-1}$ for every $k=1,...,K$;
			\FOR{every $(i,j)$}
			\STATE Compute a sub-gradient $D_t\in \partial f_{\ell_1}(U_t)$;
			\STATE Obtain $\widehat{\Theta}_{ij}$ by solving \ref{problem:lq_ij};
			\ENDFOR
		\end{algorithmic}
	\end{algorithm}

Next, we analyze the computational cost of each step of our proposed algorithm. Given $n_k$ number of samples, the sample covariance matrix $\widehat{\Sigma}_k$ can be computed in $\mathcal{O}(n_kd^2)$ time and memory (Line 3). Moreover, given each sample covariance matrix, the approximate backward mapping can be obtained by an element-wise soft-thresholding followed by a matrix inversion, which can be done in $\mathcal{O}(d^3)$ time and memory (Line 4). Finally, for each $(i,j)$ and the choices of $q=1,2$, the subproblem \ref{problem:lq_ij} can be reformulated as a linearly constrained convex quadratic problem. Suppose that $W$ has $\texttt{nnz}$ number of nonzero elements. Then, each subproblem can be solved in $\mathcal{O}(\texttt{nnz}^3)$~\cite{boyd2004convex}. Moreover, assuming that the algorithm is parallelized over $M$ machines, the total complexity of solving all subproblems is $\mathcal{O}((d^2/M)\texttt{nnz}^3)$. In the next section, we show that our proposed algorithm is extremely efficient in practice.

\section{Numerical Experiments}\label{sec:numerical}

In this section, we evaluate the performance of our proposed method on synthetically generated dataset, as well as the Glioblastoma spatial transcriptomics dataset. All experiments are implemented using MATLAB 2021b, and performed with a 3.2 GHz 8-Core AMD Ryzen 7 5800H CPU with 16 GB of RAM. We use the function \texttt{quadprog} in MATLAB to solve each subproblem.

\subsection{Synthetically generated dataset.}

First, we use synthetically generated dataset to compare the statistical performance of our proposed method with two other estimators: the fused graphical lasso (FGL) \cite{danaher2012joint} and FASJEM \cite{pmlr-v54-wang17e}. FGL is an MLE-based approach augmented by a regularizer to promote spatial similarity among different distributions. On the other hand, FASJEM uses the same Elementary Estimator \cite{Yang14} framework as ours while having different regularization term. By comparing the estimated parameters with their true counterparts, we will show that our method outperforms both FGL and FASJEM in recovering the true precision matrices. 

\vspace{1mm}

\noindent{\bf Data generation:} 
Our data generation process is motivated by ideas proposed by \citet{peng2009partial} and \citet{lyu2018condition} to imitate the gene expression profiles from a synthetically-generated co-expression network. Our main goal is to generate the data synthetically from a known distribution, and then evaluate the performance of the estimated parameters by comparing them to the ground truth.

We simulate the true precision matrices for $K$ distinct clusters (populations) with varying level of similarity. Within each cluster, we assume that the graph representing the true precision matrix has a disjoint modular structure, with power law degree distribution for nodes within each module. Specifically, we split $d$ genes into $M$ modules, with $d/M$ genes per module generated based on Barabasi-Albert model~\cite{albert2002statistical}. Within each cluster, the modules are simulated independently and concatenated to produce a block-diagonal matrix, which is treated as the true precision matrix for the corresponding cluster.

In order to simulate the true precision matrices for all clusters, we first generate a random spanning tree over clusters. Starting at the root cluster, we generate $M$ modules, and in each module, we randomly generate a graph with $d/M$ vertices according to the Barabasi-Albert model. Based on the adjacency matrix of this graph, we select the edge weights uniformly from $[-1, -0.4] \cup [0.4, 1]$. To ensure the positive semi-definiteness of the constructed precision matrix, we use 1.1 times the sum of the absolute values of all off-diagonal elements in each row as the value of the diagonal elements in that row. Finally, we construct the precision matrix as a block diagonal matrix with $M$ modules on the diagonal blocks. We then traverse the spanning tree from the root cluster and, at every new cluster, construct the precision matrix by perturbing its parent cluster. We consider two types of perturbations: $(i)$ edge weight perturbation; and $(ii)$ edge reconnection. To perform Type $(i)$ perturbation, we sample a subset of the $M$ modules at the parent cluster, and add a uniform perturbation from the interval $[-0.04,0.04]$ to the non-zero edges. For Type $(ii)$ perturbation, we replace one of the $M$ modules with a newly simulated one following a power-law degree distribution. Thus, at every cluster, the precision matrix is slightly perturbed relative to its parent, and the precision matrix differences accumulate, which means that the number of different edges of two precision matrices increases with their distance.
Figure \ref{fig:generating_adj} illustrates the precision matrices for the two adjacent clusters.
Having simulated the precision matrices, at every cluster $k$, we next collect $n_k$ samples from a zero-mean Gaussian distribution with the constructed precision matrix.

\begin{figure}[!t]
\centering
\subfloat[Parent cluster]{\includegraphics[width=2.8in]{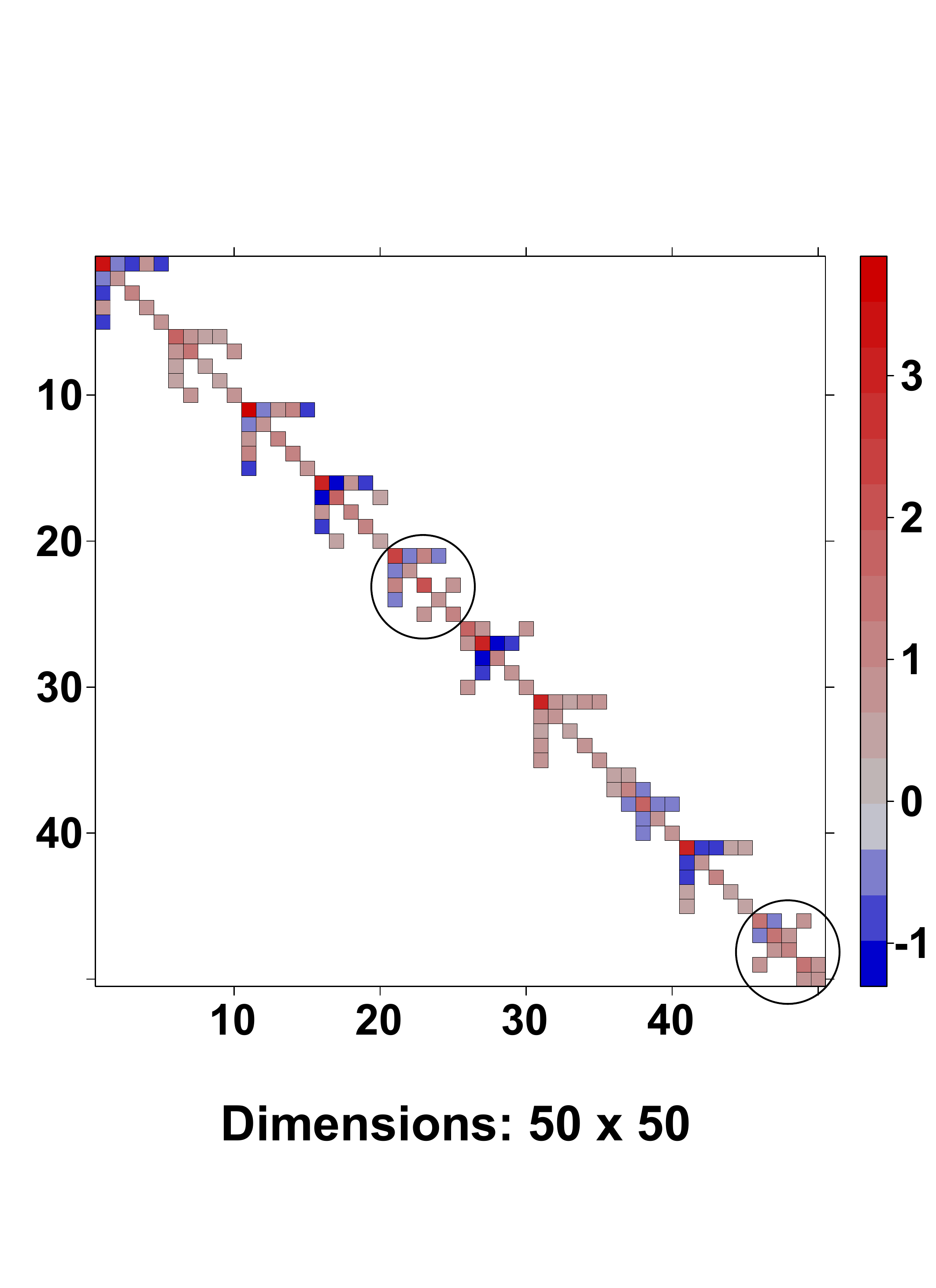}
\label{fig:adj_a}}
\subfloat[Child cluster]{\includegraphics[width=2.8in]{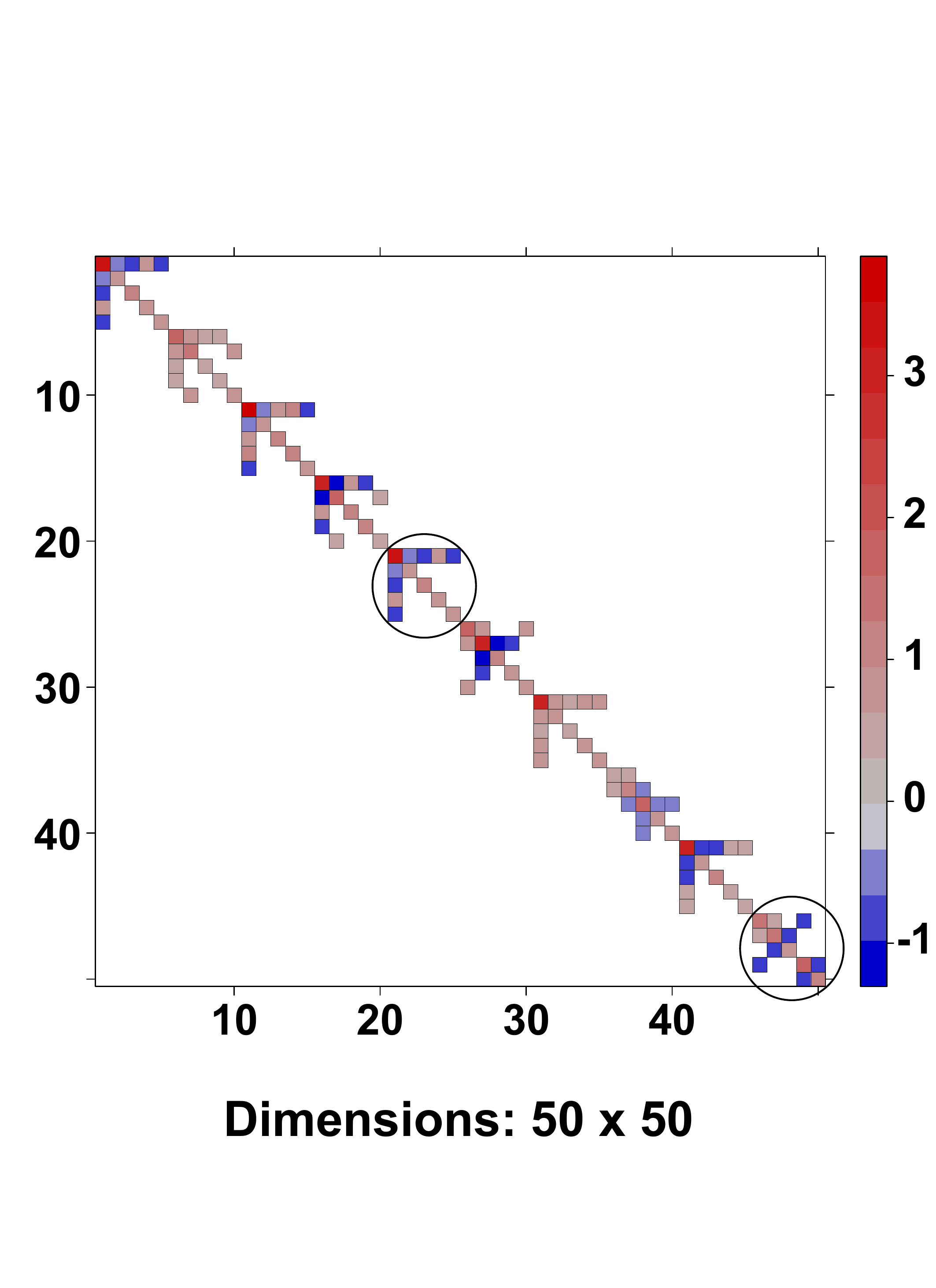}
\label{fig:adj_b}}
\caption{The child cluster is obtained from the parent cluster by regenerating module 5 and perturbing the edge weights of module 10.}
\label{fig:generating_adj}
\end{figure}

\begin{figure*}[!t]
\centering
\subfloat[Precision]{\includegraphics[width=2.2in]{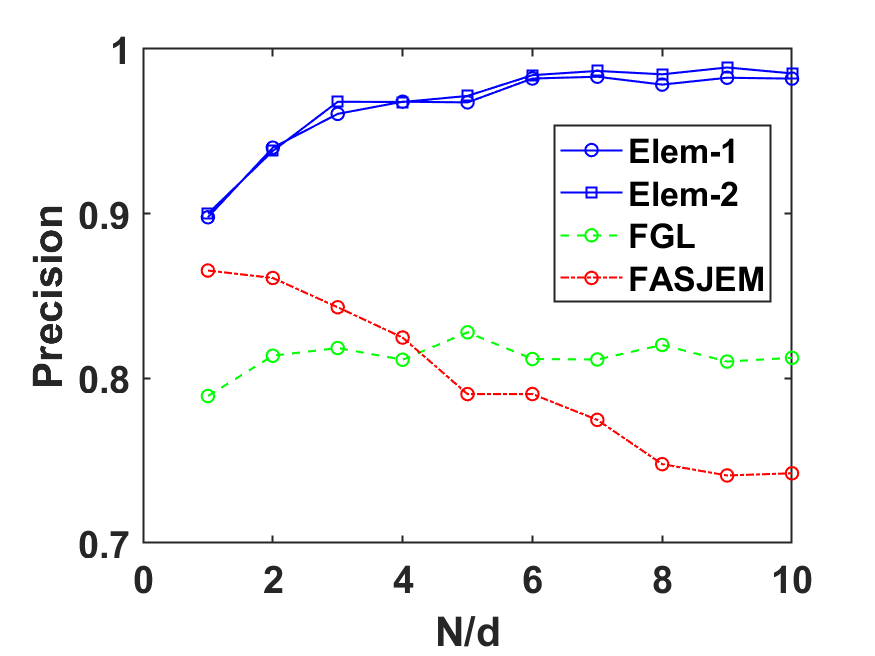}\label{fig:accuracy_a}}
\subfloat[Recall]{\includegraphics[width=2.2in]{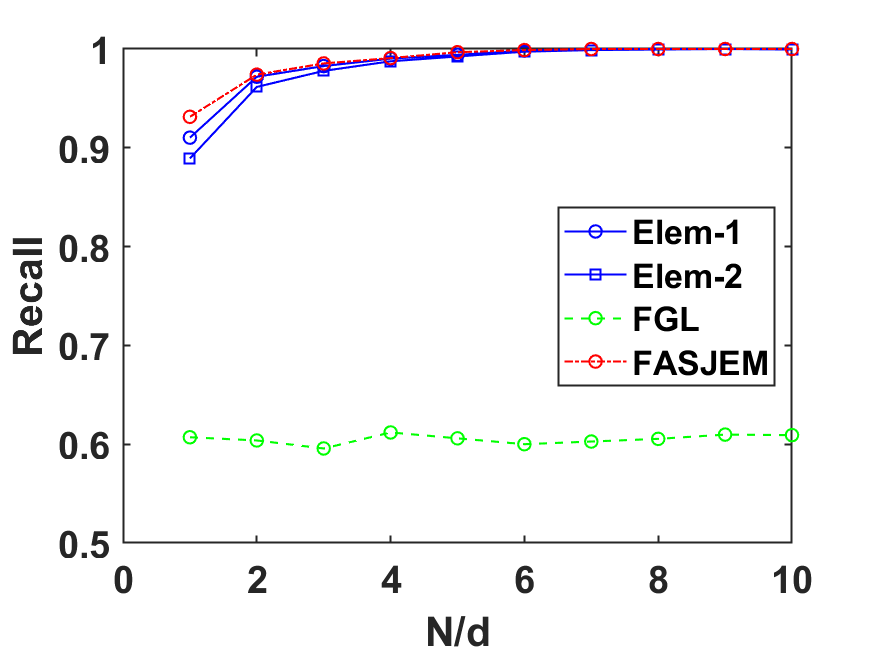}
\label{fig:accuracy_b}}
\subfloat[F1-score]{\includegraphics[width=2.2in]{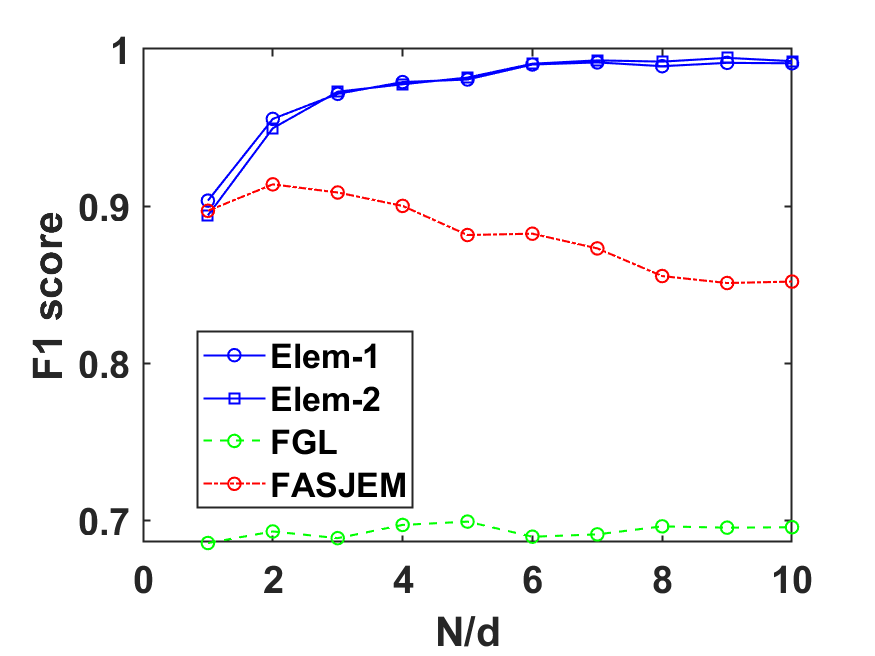}
\label{fig:accuracy_c}}
\caption{Precision, Recall, and F1-score for the estimated precision matrices for different methods with varying sample size. \textit{Elem-1} and \textit{Elem-2} perform similarly, and they both outperform \textit{FGL} and \textit{FASJEM}. The higher value of Recall for \textit{FASJEM} is due to the underestimation of the regularization parameters that promote sparsity, which in turn leads to a large value of TP.}
\label{fig:accuracy}
\end{figure*}

\vspace{1mm}
\noindent{\it \underline{Experiment 1: Varying number of samples.}} In our first experiment, we fix $K = 5, d=250$, and $M=5$, and compare the performance of \ref{problem:lq} with FGL and FASJEM with varying number of samples $n_k$.
We compare the estimation accuracy in terms of $\text{Recall} = {\text{TP}}/({\text{TP}+\text{FN}}),\ \text{Precision} = {\text{TP}}/({\text{TP}+\text{FP}}),$ and $\ \text{F1-score} = 2(\text{Recall}\times \text{Precision})/(\text{Recall}+\text{Precision})$, where TP, FN, and FP correspond to the true positive, false negative, and false positive values, respectively. To fine-tune the weight matrix $W$ and the parameters $(\mu,\gamma,\nu_k)$, we use the distance measure and BIC approach delineated in Section~\ref{sec_alg}. Moreover, we use the same BIC approach to fine-tune the parameters of FGL and FASJEM.

Figure \ref{fig:accuracy} illustrates the performance of different estimation methods. It can be seen that \ref{problem:lq} with $q=1,2$ (denoted as \textit{Elem-$1$} and \textit{Elem-$1$}) perform almost the same, and they both outperform FASJEM and FGL in terms of the Precision and F1 scores. On the other hand, the Recall score for FASJEM is artificially high due to the underestimation of the regularization parameters via BIC, which in turn leads to overly dense estimation of the precision matrices.

\vspace{1mm}
\noindent{\it \underline{Experiment 2: Varying dimension}.}
Next, we analyze the performance of our proposed method for different dimensions $d$. In particular, we consider a high-dimensional regime where $d$ is significantly larger than the number of available samples $n_k$. We fix $K=5$ and set $n_k = d/2$. The parameters $\mu,\gamma,\nu_k$ and the weight matrix $W$ are tuned as before.

Figure \ref{fig:fixL_accuracy} depicts the Precision, Recall, and F1-score, as well as the runtime of our proposed method and FASJEM with respect to $Kp = Kd(d+1)/2$ which ranges from $10^5$ to $2.5\times 10^6$.\footnote{Due to the large scale of these instances, FGL did not converge within 10 minutes even for the smallest instance with $d=200$. Therefore, it is omitted from our subsequent experiments.} It can be seen that the runtime of our proposed method scales almost linearly with $p$, with the largest instance solved in less than 2 minutes. On the contrary, FASJEM has an undesirable dependency on $p$, with a runtime exceeding 10 minutes for medium-scale instances of the problem. The linear time of our algorithm with respect to $p$ is due to its decomposable nature of over different coordinates of the precision matrices. 

\vspace{1mm}

\begin{figure*}[!t]
\centering
\subfloat[Precision]{\includegraphics[width=3in]{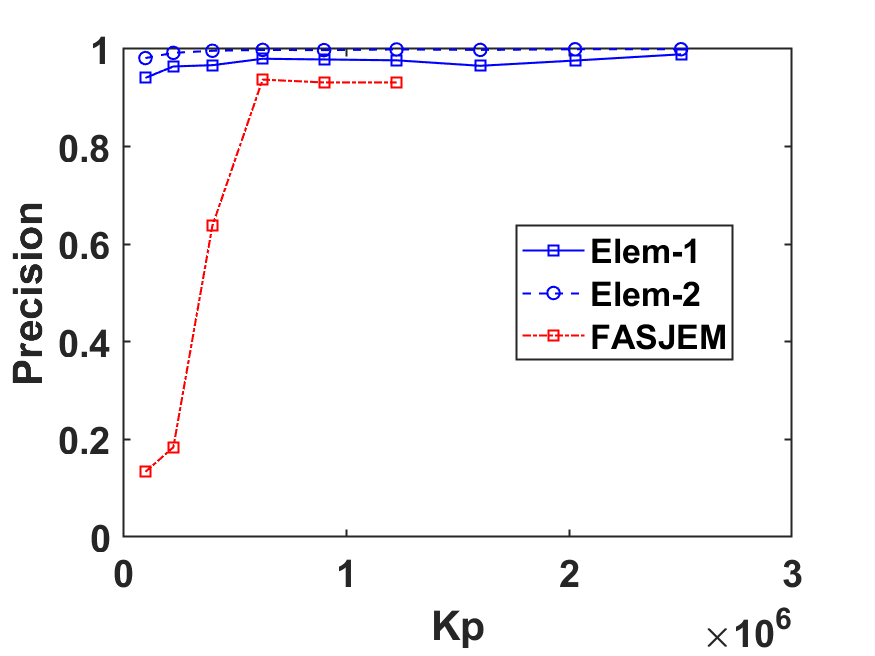}
\label{fig:fixL_a}}
\hspace{-6mm}
\subfloat[Recall]{\includegraphics[width=3in]{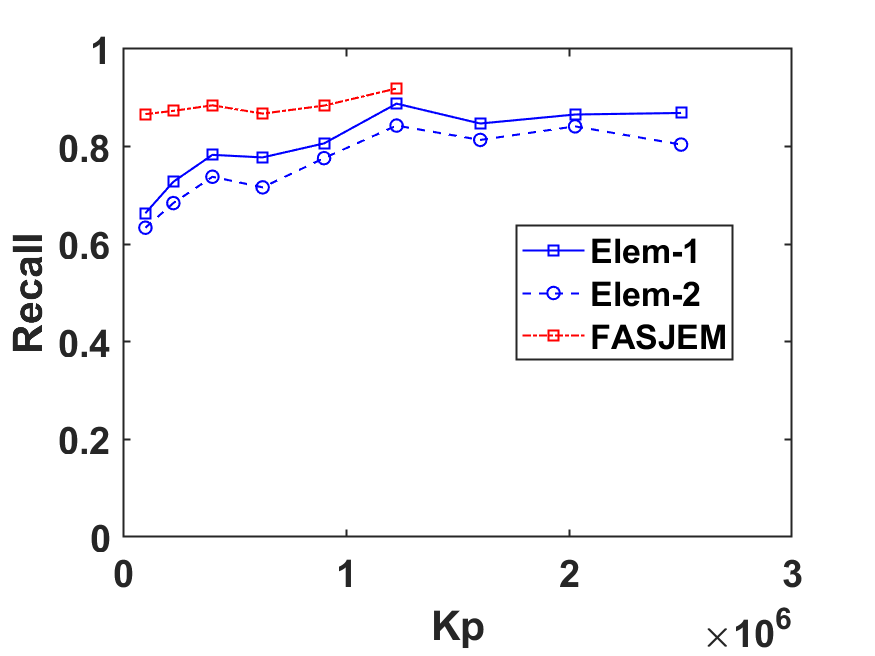}
\label{fig:fixL_b}}
\hspace{-6mm}
\subfloat[F1-score]{\includegraphics[width=3in]{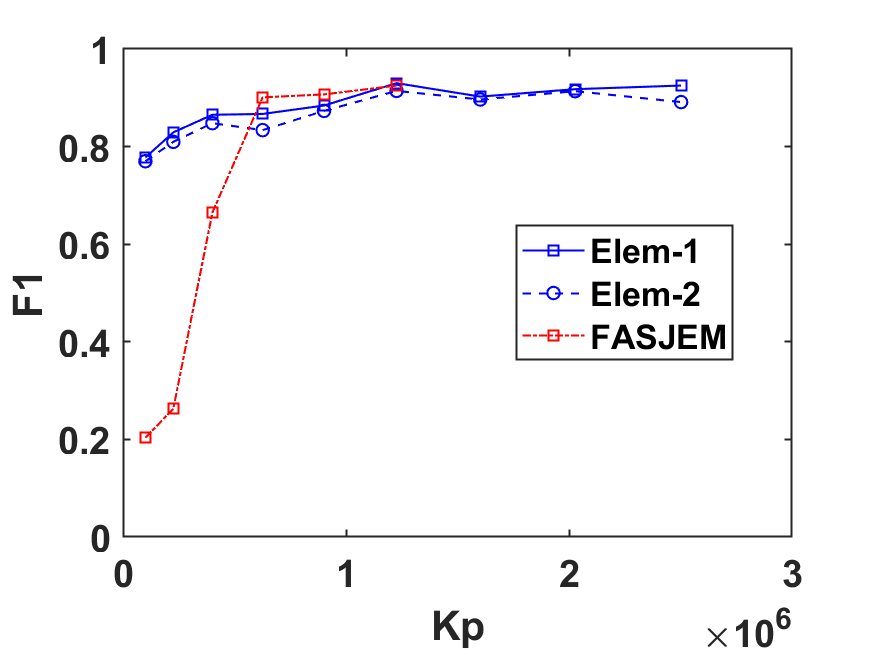}
\label{fig:fixL_c}}
\hspace{-6mm}
\subfloat[Runtime]{\includegraphics[width=3in]{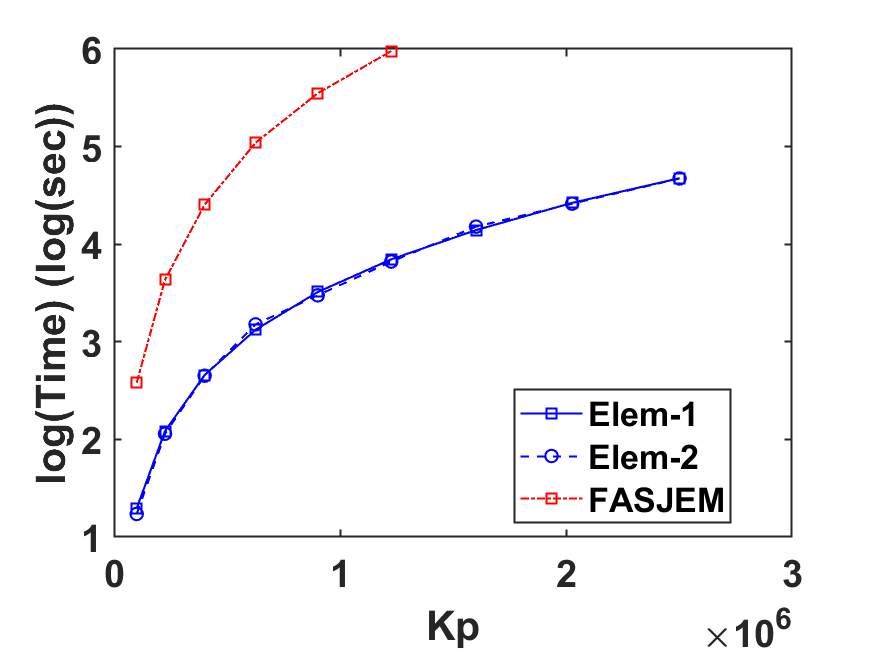}
\label{fig:fixL_d}}
\caption{Precision, Recall, F1-score for the estimated precision matrices, as well as the runtime of our proposed method with varying dimension. \textit{Elem-1} and \textit{Elem-2} outperform \textit{FASJEM} in terms of Precision and F1-score, while \textit{FASJEM} outperforming \textit{Elem-1} and \textit{Elem-2} in terms Recall. Similar to the previous experiment, the higher value of Recall for \textit{FASJEM} is due to the underestimation of the regularization parameters which lead to overly dense precision matrices. Moreover, both \textit{Elem-1} and \textit{Elem-2} are drastically faster than \textit{FASJEM}.}
\label{fig:fixL_accuracy}
\end{figure*}

\begin{figure*}[!t]
\centering
\subfloat[Precision]{\includegraphics[width=3in]{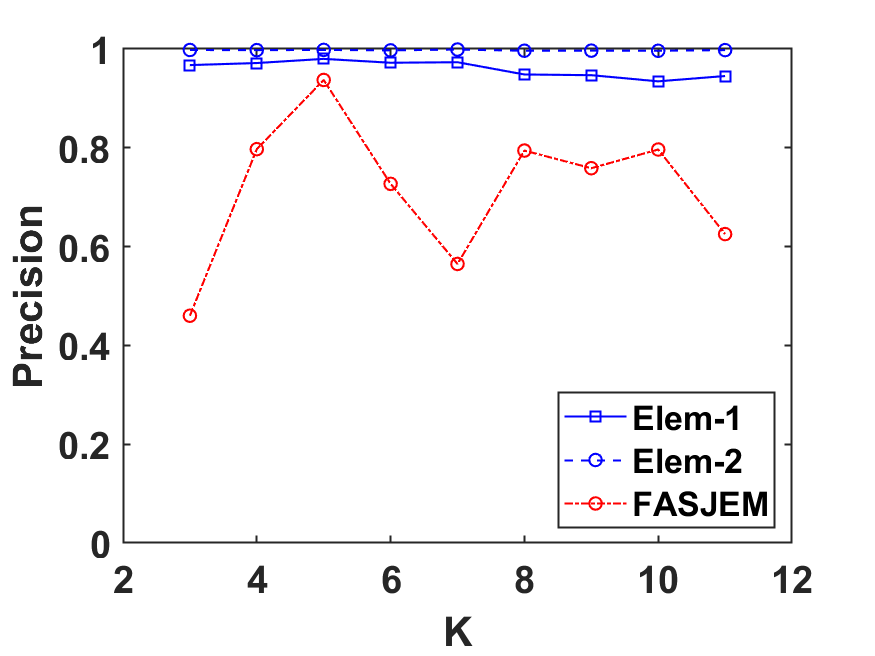}
\label{fig:fixP_a}}
\hspace{-6mm}
\subfloat[Recall]{\includegraphics[width=3in]{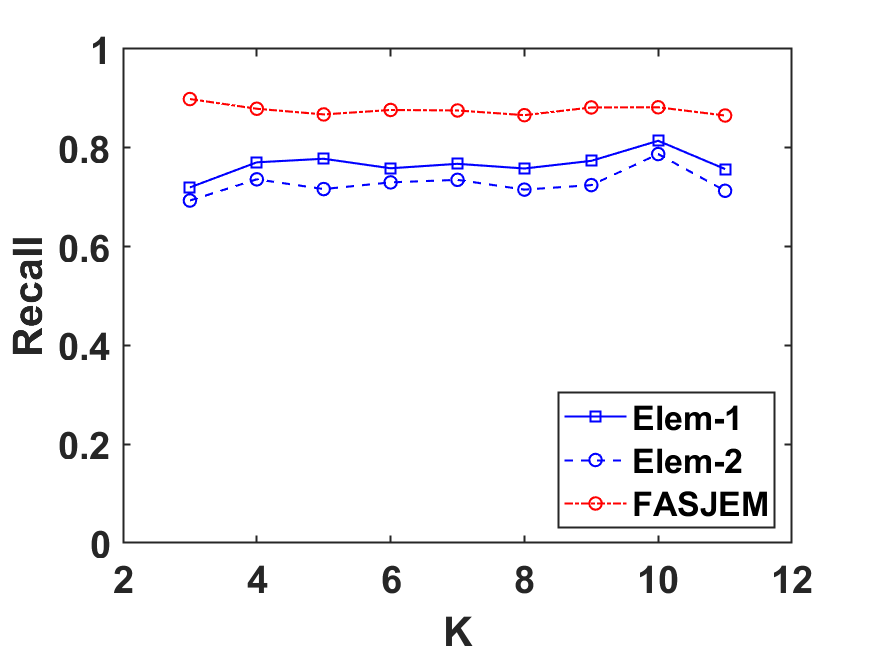}
\label{fig:fixP_b}}
\hspace{-6mm}
\subfloat[F1-score]{\includegraphics[width=3in]{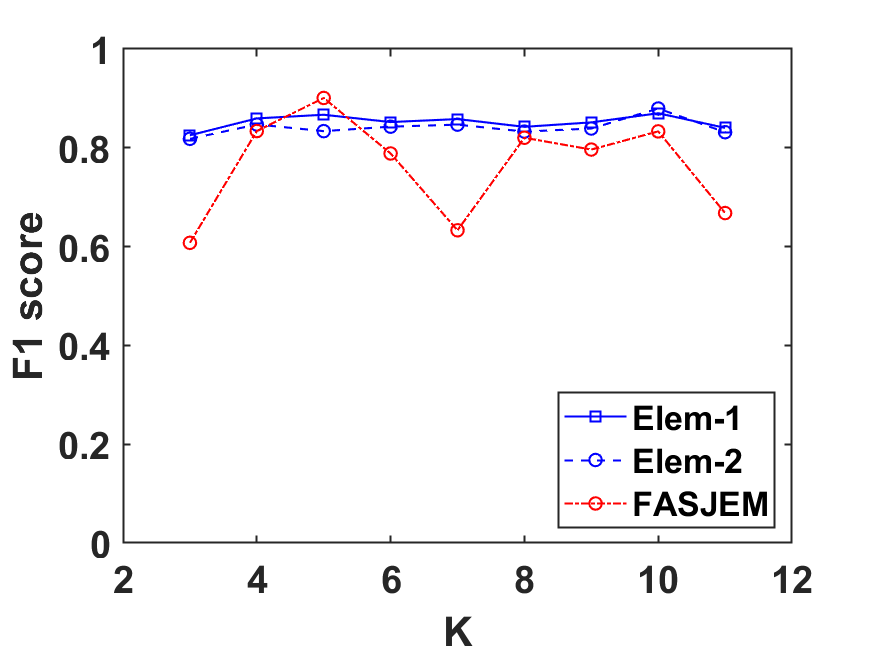}
\label{fig:fixP_c}}
\hspace{-6mm}
\subfloat[Runtime]{\includegraphics[width=3in]{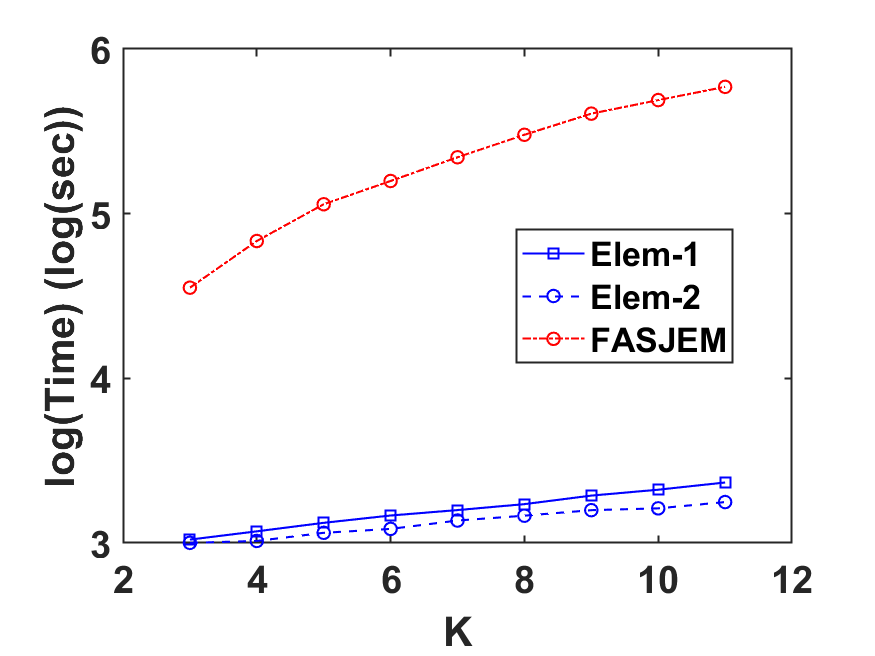}
\label{fig:fixP_b}}
\caption{Precision, Recall, F1-score for the estimated precision matrices, as well as the runtime of our proposed method with varying number of clusters. \textit{Elem-1} and \textit{Elem-2} outperform \textit{FASJEM} in terms of Precision and F1-score, while \textit{FASJEM} outperforming \textit{Elem-1} and \textit{Elem-2} in terms Recall. Similar to the previous experiments, the higher value of Recall for \textit{FASJEM} is due to the underestimation of the regularization parameters which lead to overly dense precision matrices. Moreover, both \textit{Elem-1} and \textit{Elem-2} are drastically faster than \textit{FASJEM}.}
\label{fig:fixP_accuracy}
\end{figure*}
\noindent{\it \underline{Experiment 3: Varying number of clusters.}}
Finally, we evaluate the performance of our method with varying number of clusters $K$. We fix $d=500$, $M=10$ and $n_k = 250$, and use the same tuned parameters in the previous experiment. Figure \ref{fig:fixP_accuracy} shows the Precision, Recall, and F1 score for our proposed method and FASJEM, as well as their runtime with respect to $K$. Similar to the previous experiments, both Elem-1 and Elem-2 outperform FASJEM in terms of the estimation accuracy. Moreover, it can be seen that in practice, the runtime of Elem-1 and Elem-2 scale almost linearly with $K$.

\section{Application to Glioblastoma Spatial transcriptomics dataset}

 We collected gene expression profile using the Visium spatial transcriptomics (ST) platform from a primary GBM patient tumor showing high perfusion signal in diffusion MRI (relative cerebral blood volume parameter derived from dynamic susceptibility contrast MR perfusion). We sampled two adjacent tissue sections, giving us ~6500 spots with transcriptomics data from this region. Since most routine clustering algorithms for spatial transcriptomics datasets are only based on expression-based proximity between cells, and completely ignore the spatial information, we first define a simple clustering algorithm that is also informed by spatial context.

We integrate data from adjacent tissue slices using the reciprocal PCA method in the \texttt{R} package for single cell data analysis (also known as \textit{Seurat}) \cite{hao2021integrated}. We use the dimension reduction algorithm PHATE \cite{moon2019visualizing} to obtain a 3D embedding of the integrated counts data. We then compute pairwise Euclidean distances between spots in the embedded space, and with their spatial coordinates. We perform upper quantile normalization of distance matrices based on their 75th quantile to ensure that both expression and spatial distances are in the same scale, and use their sum to define pairwise distances between spots. This dissimilarity matrix is used as input for PAM clustering. Optimal number of clusters (k = 5) is identified using the Calinski-Harabasz criterion \cite{calinski1974dendrite}, with the resulting clusters shown in  Figure \ref{fig:hp_clusters}. 

\begin{figure}[!t]
    \centering
    \includegraphics[width=6in]{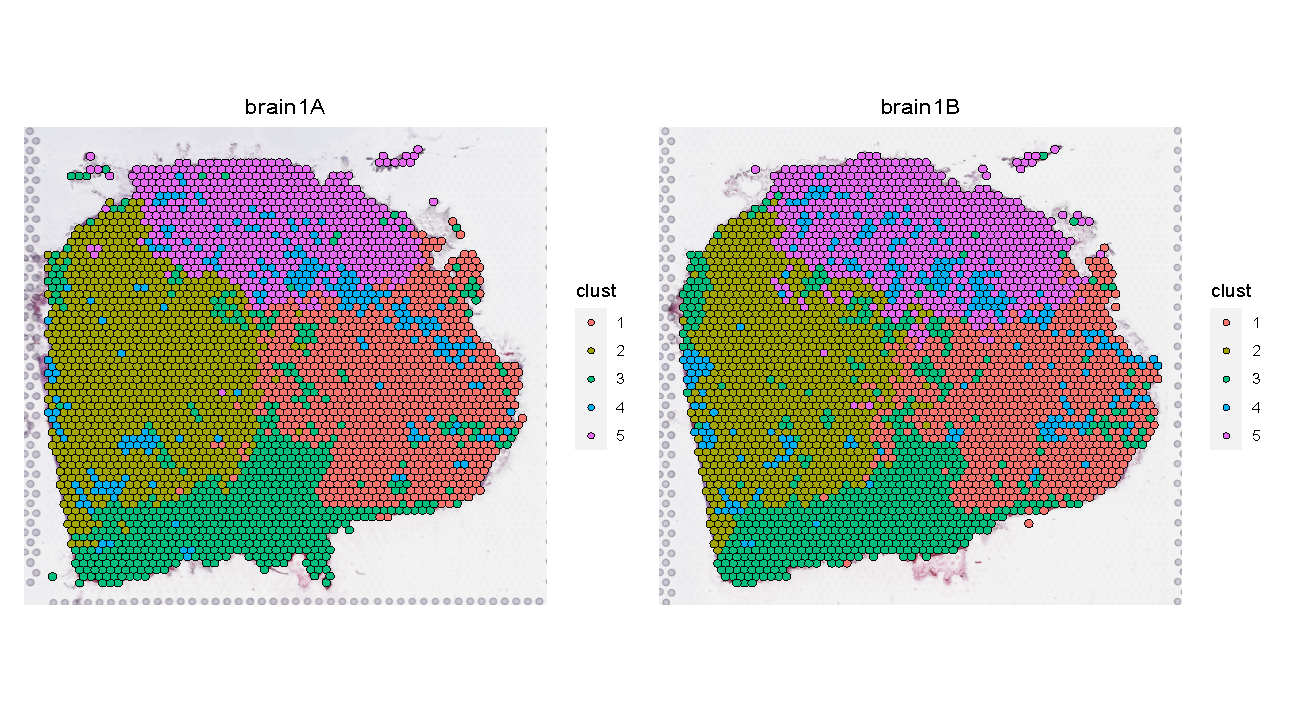}
    \caption{Adjacent tumor sections from a primary GBM patient sample are separated into five distinct clusters informed by their expression similarity and spatial proximity.}
    \label{fig:hp_clusters}
\end{figure}

In order to understand biological characteristics of these clusters, and to aid in downstream interpretation of inferred networks, we performed spot deconvolution using the RCTD algorithm \cite{cable2021robust}. Since spots in the Visium microarray have a resolution of about 60$\mu m$, they could be composed of multiple cell types. We thus used annotated single cell RNASeq dataset from \cite{darmanis2017single} to identify cell type compositional differences between the regions. We visualize in Figure \ref{fig:hp_escore} the proportion of each spot containing each of the major cell types.  We see that the tissue is primarily composed of neoplastic cells, with some vascular niches and Astrocytic populations.  We can see that Cluster 4 corresponds to a distinct vascular niche in the tumor with significant immune infiltration, and Cluster 5 has some non-tumor astrocytic cells. Thus the obtained clusters are biologically meaningful, and we can now seek to understand how gene network interactions vary in different microenvironments  of this tumor.

\begin{figure}[!t]
    \centering
    \includegraphics[width=1\textwidth]{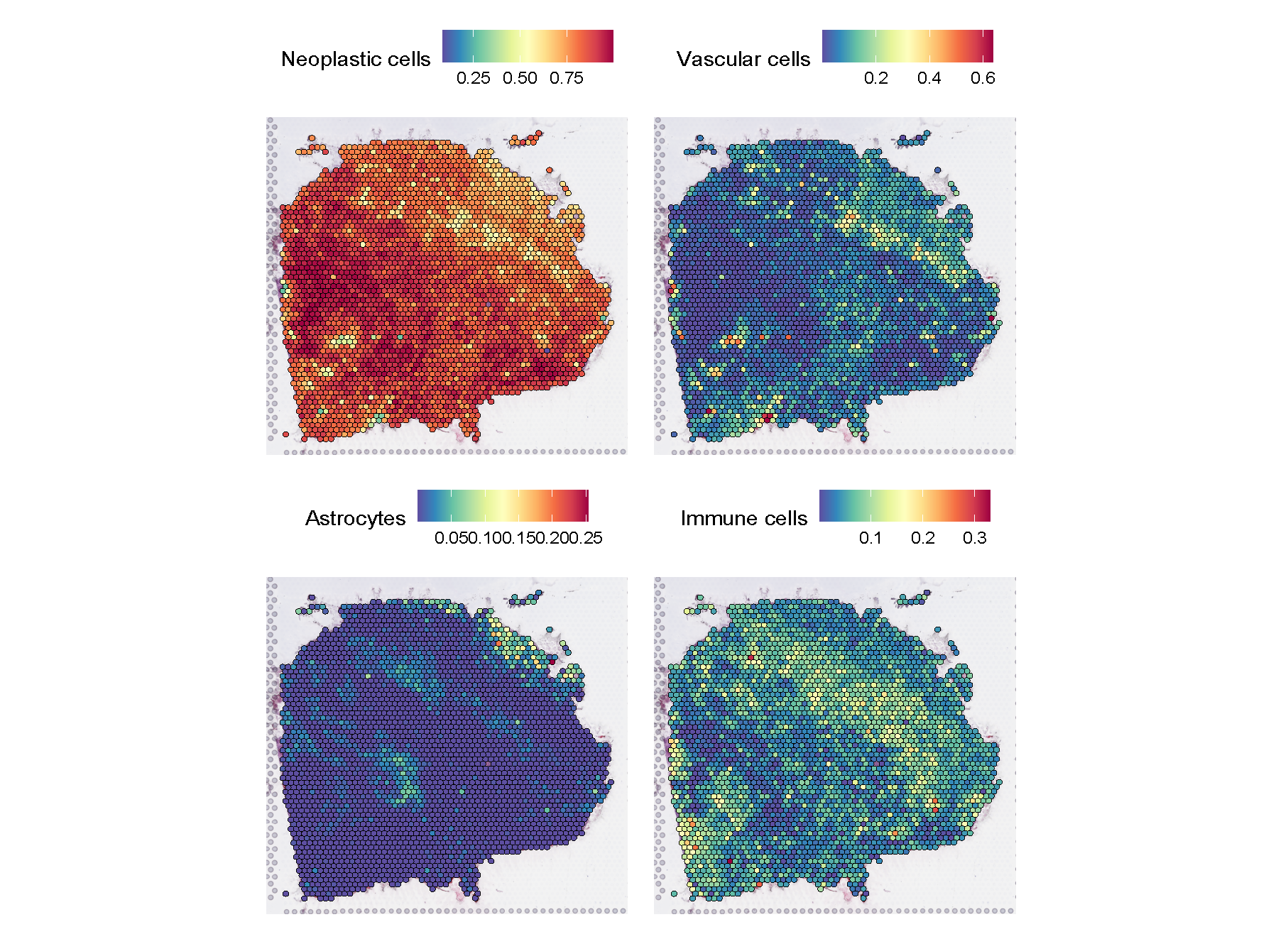}
    \caption{Output of RCTD spot deconvolution algorithm, visualized as fraction of spots composed of major contributing cell types. Cluster 4 is enriched for Vascular and immune cells, and cluster 5 has some nascent astrocytic populations.}
    \label{fig:hp_escore}
\end{figure}

Having defined biologically meaningful subdomains in the tumor section, we identify the top 2500 genes showing significant spatial trends in their expression, as determined using the SparkX algorithm \cite{zhu2021spark}. Since the Visium data is highly sparse, we carry out a non-paranormal transformation of normalized spot-level counts data using the \texttt{huge} package in \texttt{R} \cite{zhao2012huge}. Inter-cluster similarity constraints for network inference are imposed based on pairwise distances between the cluster medoids. Similar to the previous case study, we use the BIC criterion to learn the optimal parameter values.
Using these parameters, we find that of 2500 genes, 1180 have an edge in at least one cluster.

The number of inferred edges per cluster are respectively 4511, 13785, 446, 8400 and 4534. The significant variation in the number of edges in the clusters could be caused by the difference in numbers of detected read counts across the tissue section. This is shown in Figure \ref{fig:nFeature}, where we can see that Cluster 3 has significantly fewer detected genes than the other regions. While performing count imputation could help to some extent, we observed that the variation in the expressed gene counts persisted, showing similar spatial trends. This is therefore a biological and not technical artefact of the data.

\begin{figure}[!t]
    \centering
    \includegraphics[width=6in]{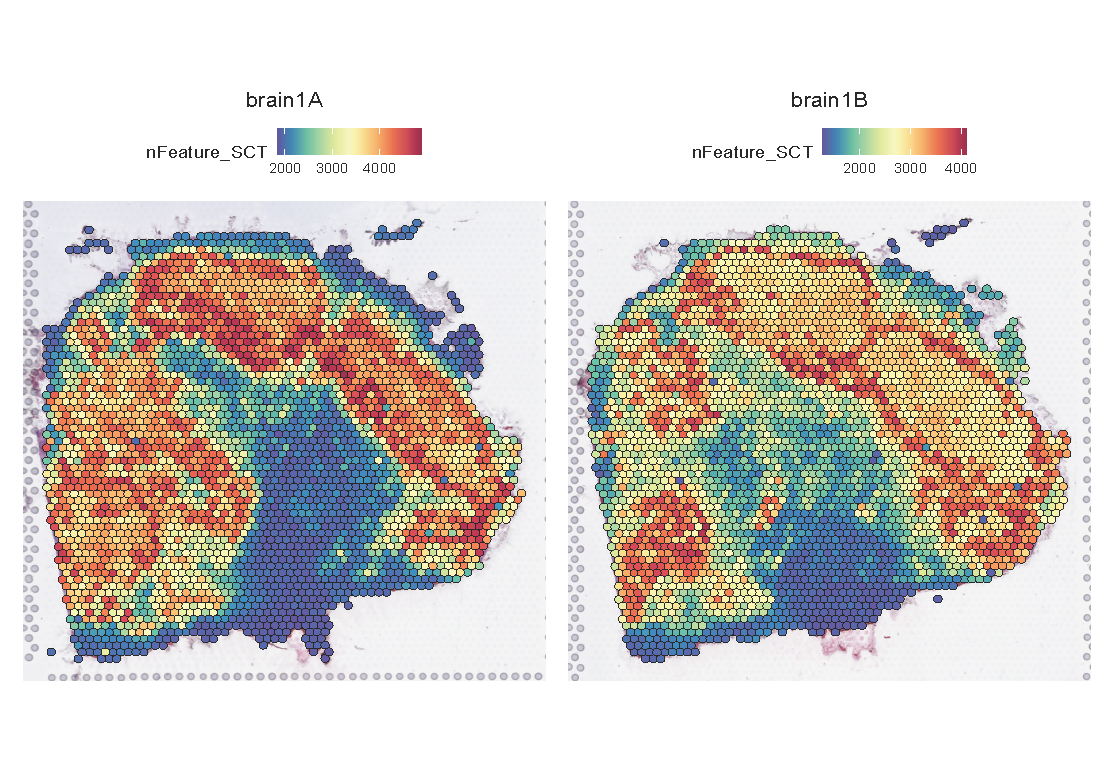}
    \caption{The number of detected reads and genes varies significantly across the tissue section. The regions associated with cluster 3 have the least number of detected genes, and could explain the very low number of edges in the corresponding graph.}
    \label{fig:nFeature}
\end{figure}

To compare the extent of similarity between the inferred networks, we use the DeltaCon algorithm \cite{koutra2013deltacon}, a statistically principled and scalable inter-graph similarity function. The relative similarity values are shown in Figure \ref{fig:deltacon}. We can see that inferred networks from the different clusters are largely distinct, with Clusters 1 and 3 having the maximum pairwise similarity of 0.27. The network in Cluster 4 is maximally different from the other regions. This is as expected, given that this region is compositionally most unique. The gene network in Cluster 2 is also very different from other zones. As we saw in Figure \ref{fig:nFeature}, this region is transcriptionally most active, and the inferred graph has nearly twice as many interactions between genes as in the other regions.

\begin{figure}[!t]
    \centering
    \includegraphics[width=0.6\textwidth]{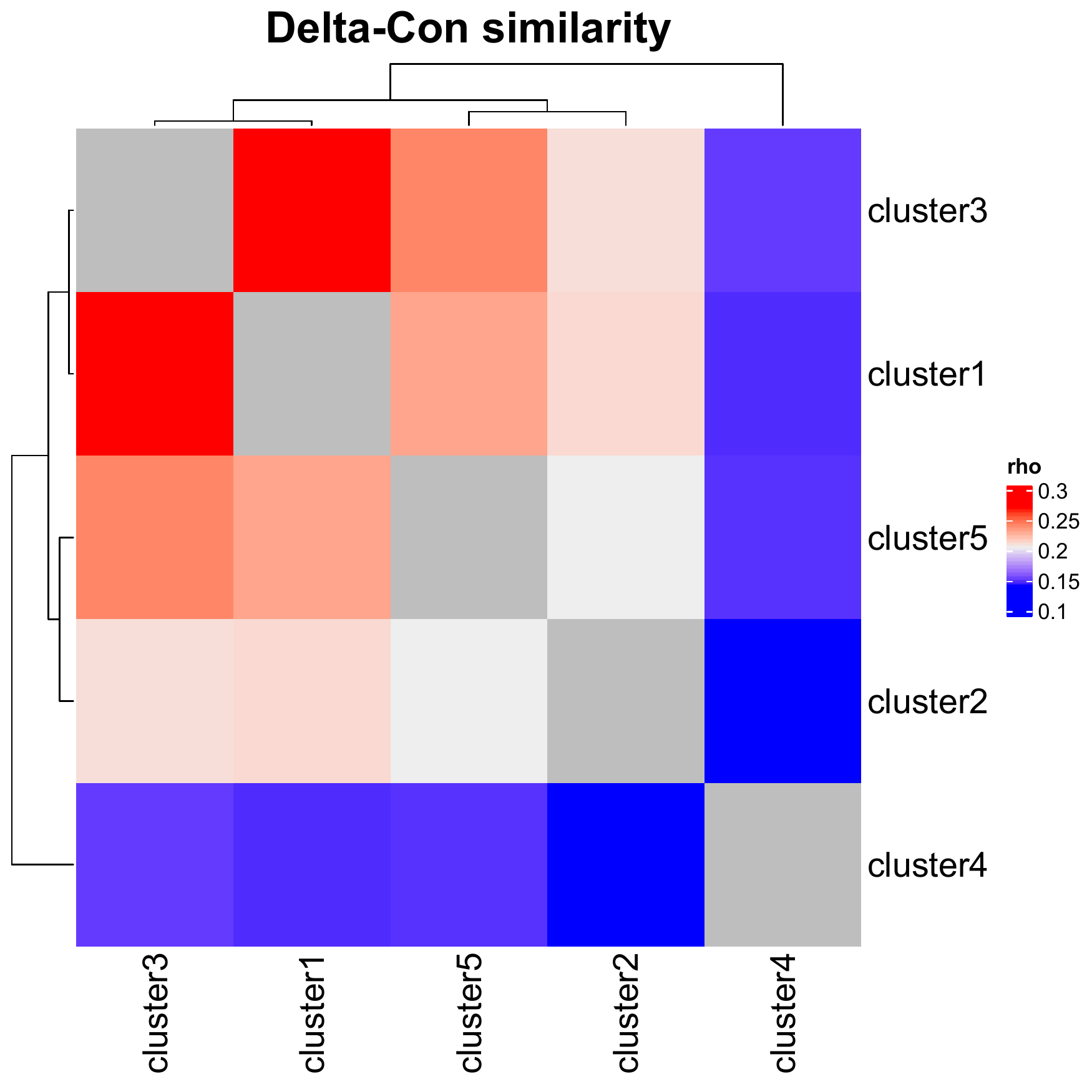}
    \caption{Inter-network similarity computed using the Delta-Con algorithm highlight extent of change in gene regulatory networks between different clusters.}
    \label{fig:deltacon}
\end{figure}

Figure \ref{fig:hp_networks} shows the top 15 strongest edges in each cluster from the inferred network, with node sizes scaled by their respective degree and negative interactions shown as dotted lines. We can immediately see that each cluster has distinct underlying regulatory interactions driving their transcriptional states, even if they appear compositionally homogeneous.

\begin{figure}[!t]
    \centering
    \includegraphics[width=6.5in]{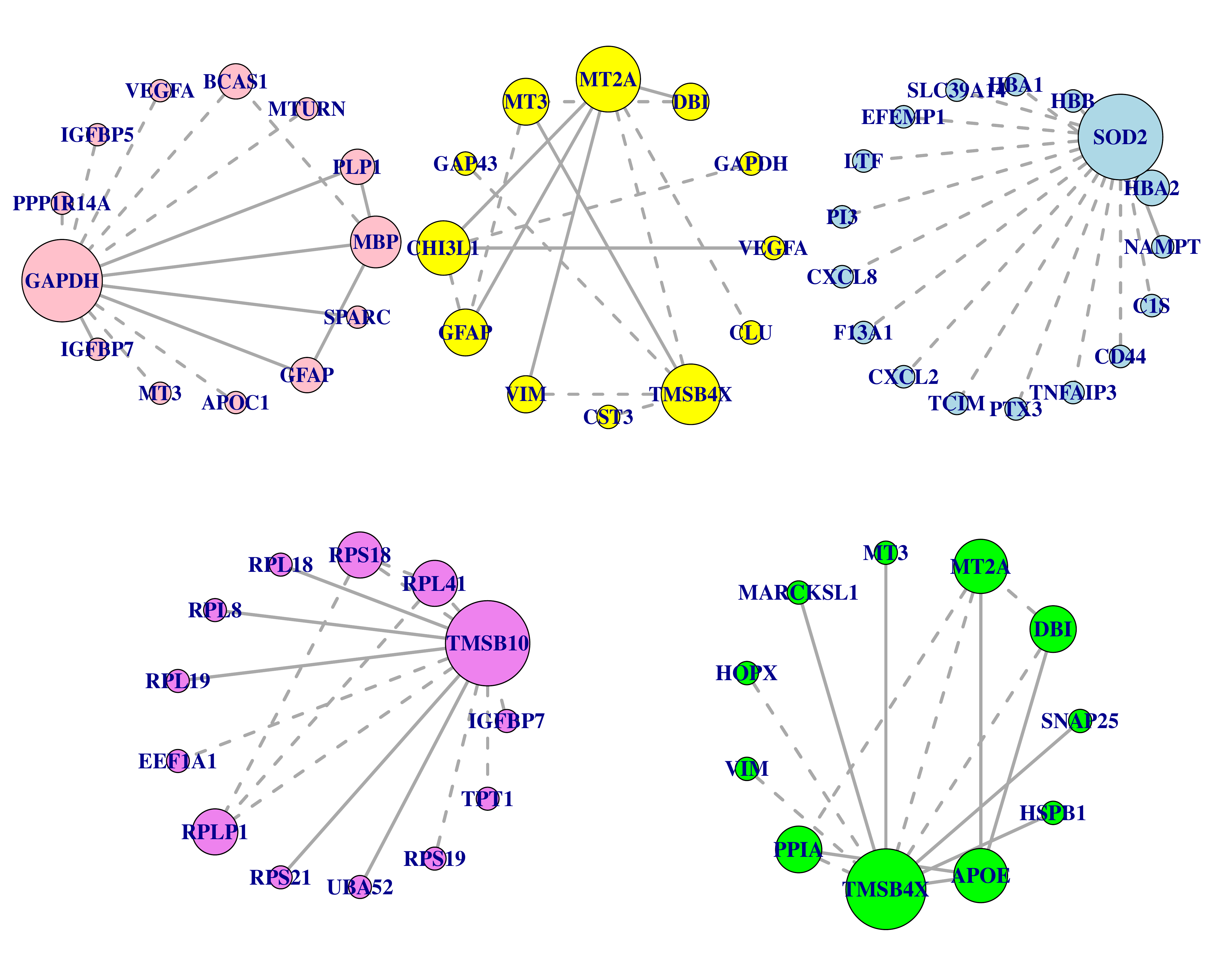}
    \caption{Top 15 strongest edges in each network are shown. Size of the nodes reflects their degree. Dotted lines indicate negative interactions. The nodes are colored by the cluster identity. }
    \label{fig:hp_networks}
\end{figure}

Transcription factors are proteins that play a dominant role in regulating gene expression networks of cells and are particularly important in driving tumor growth and evolution \cite{bushweller2019targeting}. By binding to regulatory regions of target genes, they are responsible for enhancing or suppressing gene expression and thereby controlling cell states. Regulatory interactions involving TFs are therefore of particular interest in understanding gene networks. We highlight these interactions in Figure \ref{fig:TFs}. Since Cluster 2 has an order of magnitude more edges, we highlight only the top 100 edges.

Frequently highlighted TFs active across different regions include the AP1 family TFs FOS and JUN, which are known downstream effectors of the Mesenchymal state in Gliomas \cite{marques2021nf1}, and other master regulators such as CEBPD \cite{wang2021ccaat} and oligodendroglial lineage factors OLIG1 and OLIG2 \cite{kosty2017harnessing}. We also see significant activity of SOX2 in Cluster 1, a known drivers of stemness features and radiation-resistance in Gliomas \cite{stevanovic2021sox}. Cluster 3 shows significant activity of Lactotransferrin (LTF), which encodes an iron-binding protein with known innate immune and tumor-suppresive activity \cite{cutone2020lactoferrin}. Interestingly, this gene has also been characterized as being an upstream master regulator of different GBM subtypes \cite{bozdag2014master}, warranting further exploration of this gene in driving tumorigenesis in GBM. 

The TF network in Cluster 4, which represents the peri-vascular niche, is most different from the other regions, as expected given its unique microenvironment. This region shows prominent activity of HES4, a known downstream effector of the NOTCH signaling pathway that is known to inhibit cell differentiation and helps maintain the stemness features in Gliomas \cite{bazzoni2019role}. HES4 specifically regulates proliferative properties of neural stem cells, and reduces their differentiation. This is a very promising observation, given that the perivascular niche is known to harbor therapy-resistant glioma stem cells whose properties are critically driven through NOTCH signaling \cite{jung2021tumor}. Cluster 5 has dominant activity of NME2, an nucleotide diphosphate kinase enzyme involved in cellular nucleotide metabolism and DNA repair \cite{puts2018nuclear}. The NME2 protein has also been identified to be a highly specific Tumor-associated antigen in IDH mutant Gliomas \cite{dettling2018identification}. By studying the regulatory networks in each cluster separately, we are thus able to infer differential activity of different master regulators in different microenvironmental niches, which informs us of the varying cell states.

\begin{figure}[!t]
    \centering
    \includegraphics[width=6.5in]{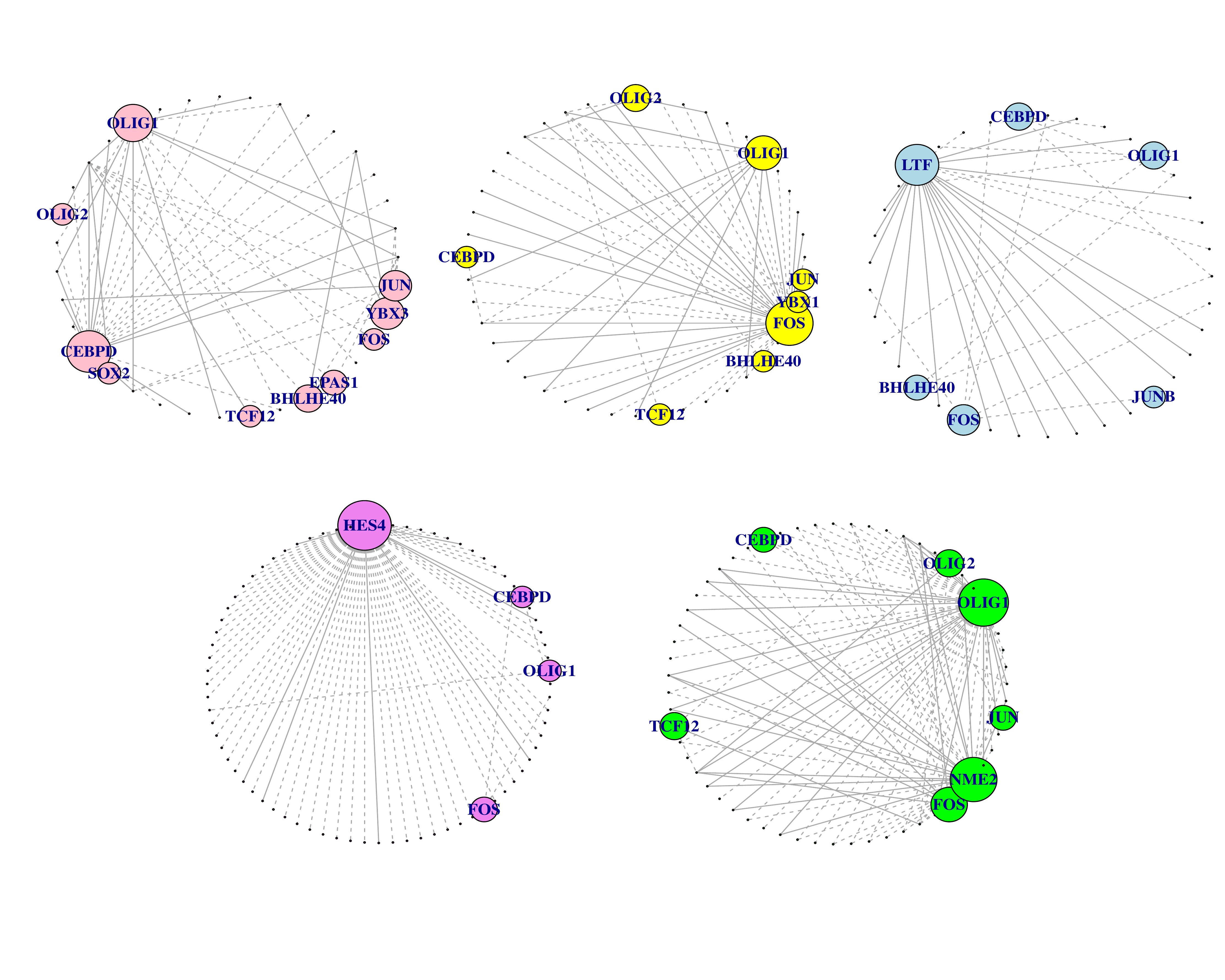}
    \caption{Transcription factor interactions in each cluster are highlighted. For each cluster, we show the top 100 strongest edges involving TFs. The graphs are visualized using the Davidson-Harel layout. Color scheme is same as in Figure \ref{fig:hp_networks}.}
    \label{fig:TFs}
\end{figure}

We observed that the majority of the estimated edges are unique to the respective clusters, with Cluster 4 being most different with $84\%$ unique edges. This network is characterized by a very high degree of connectivity between genes encoding ribosomal proteins. About a third of the nodes (60 / 184) having node degree over 100, and is qualitatively very different from the degree distribution we saw in the other networks. We next characterized the networks using different centrality measures such as their degree, betweenness, closeness, eigen and pagerank, each of which measures a different aspect of importance of nodes \cite{du2019social}. We reduce each network to its set of unique edges, and consider the top $10\%$ of nodes by each centrality measure to be hub genes. We then compare multiplicity of these nodes across the different networks. As expected, majority of hub genes are specific to individual networks. However in spite of removing  shared edges, we observe that the gene TMSB4X is identified as a shared hub gene in Clusters 1, 2 and 5, suggesting its potential importance in Glioma growth. We also note that TMSB4X has the highest degree of all genes in a base network that is shared across clusters. This gene encodes an actin-sequestering polypeptide, and is known to promote stemness and invasive phenotype in Gliomas \cite{cheng2020identification}.

Next we performed a Gene Ontology Enrichment analysis with the cluster-specific hub genes shown in Figure \ref{fig:sp_hubs}. We can clearly see that the networks in each cluster are specific for different tasks. Cluster 1 shows an enrichment for neural differentiation related genes, which are known to be downregulated in GBM, as well as extracellular matrix related processes associated with invasive properties of the tumor. Cluster 2 specific hubs are associated with metabolic and biosynthetic processes. Cluster 3 hubs are associated with innate immune responses, in agreement with our observation that LTF is a major TF in this network. Cluster 4 has a large number of ribosomal genes with high connectivity. High levels of ribosomal protein activity has been shown to be associated with promoting stemness characteristics of Gliomas \cite{shirakawa2020ribosomal}, and we see it to be a defining characteristic of the peri-vascular niche. Cluster 5 shows enrichment for neuronal processes like synaptic transmission, consistent with the presence of normal astrocytes in this region.

In summary, using our scalable framework of inferring gene regulatory networks across multiple spatially informed clusters, we are able to learn and characterize variations in tumor cell states with their microenvironmental niches and identify master regulators that are differentially active in each region. We are able to reinforce the role of known TFs such as CEBPD, FOS-JUN, OLIG and SOX family TFs, as well as identify less known drivers of context-dependent tumor adaptation. Using our method, we can demonstrate the significant strengths and show how it can be used to get a very deep understanding of tumor growth and adaptation from spatial gene expression datasets.

\section{Conclusion}

In this work, we study the inference of spatially-varying Gaussian Markov random fields (SV-GMRFs) and its application in gene regulatory networks in Spatially resolved transcriptomics. The existing methods for inferring GMRFs suffer from the so-called \textit{curse of dimentionality}, which limit their applicability to small-scale and spatially-invariant networks. To address this challenge, we propose a simple and efficient inference framework for inferring SV-GMRFs that comes equipped with strong statistical guarantees. Contrary to the existing MLE-based methods, our proposed method is amenable to parallelization and is based on solving a series of decomposable convex quadratic programs. We show that our proposed method is extremely efficient in practice, and outperforms the existing state-of-the-art techniques---both computationally and statistically. We study the developed framework in the context of inferring gene networks underlying oncogenesis, using Glioblastoma (GBM) as a case study. We uncover the nature of spatial gene relationships across multiple subregions of the tumor. Given that the tumor cell states are dynamically regulated by their spatial context, our discovered context specific master regulators is an important step towards developing targeted therapies in the future.

\begin{figure}[!t]
    \centering
    \includegraphics[width=6in]{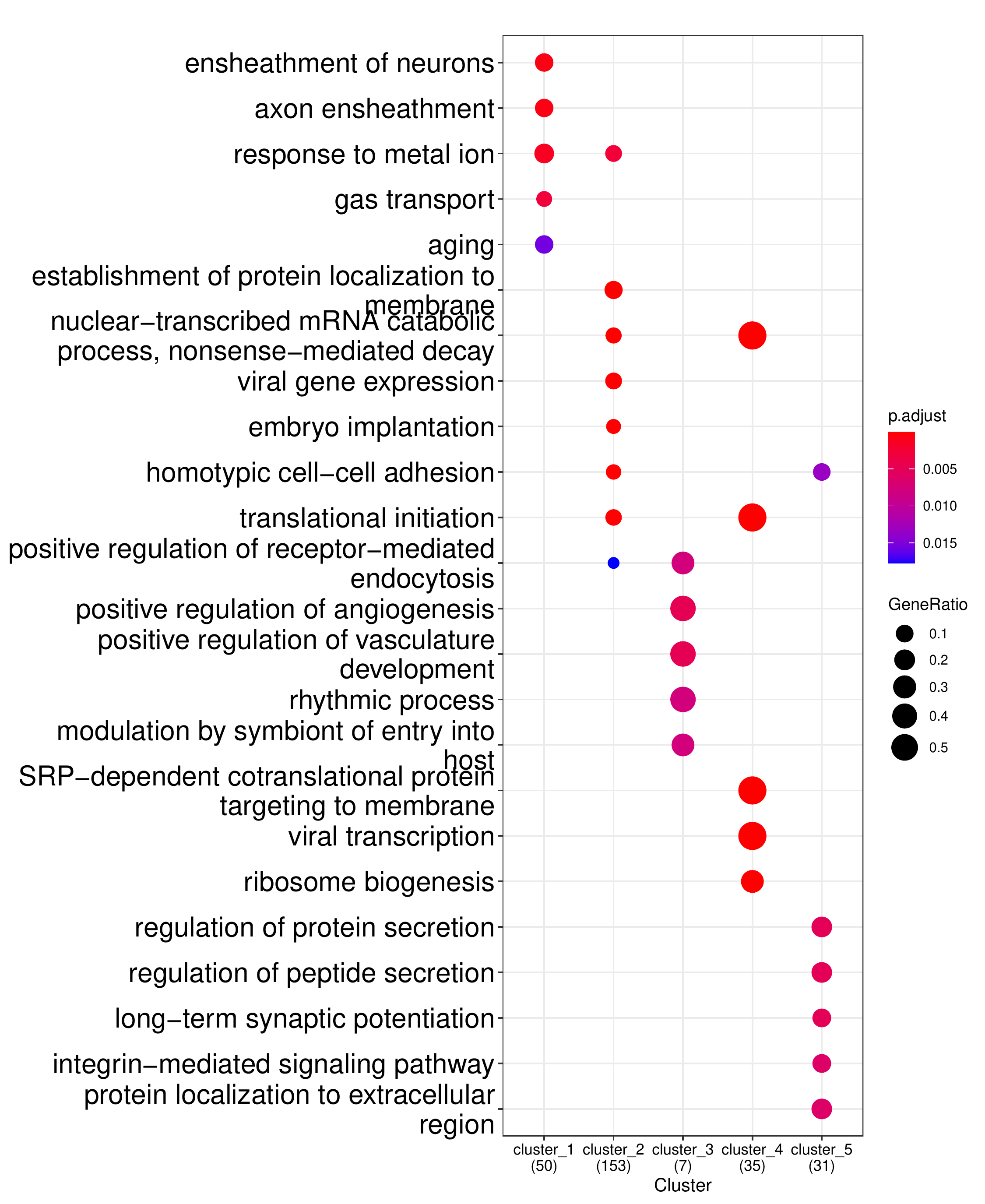}
    \caption{GO enrichment analysis for cluster-specific hub genes shows specific activity of different biological processes across the tumor section. Cluster 4 that represents a perivascular niche shows a high level of activity of ribosomal genes and translational activity.}
    \label{fig:sp_hubs}
\end{figure}

\newpage
\appendix

\section{Proof of Theorem~\ref{thm_smooth}}
The overarching idea behind the proof of Theorem~\ref{thm_smooth} is as follows: we first derive a deterministic guarantee on the estimation error of the proposed method. We then prove our main theorem by extending this result using probabilistic concentration bound.

Our first result provides a set of deterministic conditions for the optimal solution to have small estimation error and correct sparsity pattern.

\begin{prop}[Deterministic guarantee for smoothly-changing SV-MRFs]
	For a given index $(i,j)$, suppose that the regularization parameters satisfy 
	\begin{align}
	    &\gamma< \frac{1}{2K\norm{W}_{\max}},\label{eq_1_smooth}\\
	    &6\|\tilde{F}^*_{ij}\!-\!\Theta^\star_{ij}\|_{\infty}\!+\!6\gamma K\norm{W}_{\max} D\leq \mu< \frac{3\Theta_{\min}}{2}.\label{eq_2_smooth}
	\end{align}
	Then, the following statements hold:\\
	\begin{itemize}
	    \item (Sparsistency) The solution $\widehat{\Theta}_{ij}$ is unique and satisfies $\supp(\widehat{\Theta}_{ij}) = \supp({\Theta}^\star_{ij})$.
	    \item (Estimation error) We have $\|{\widehat{\Theta}_{ij} -\Theta_{ij}^\star}\|_{\infty} \leq { 2\mu}/{3}$
	\end{itemize}
	\label{thm_smooth_det}
\end{prop}
The proof of the above proposition is provided in Appendix~\ref{proof_cor_smooth_det}. Next, we show how this proposition can be used to complete the proof of Theorem~\ref{thm_smooth}. In order to invoke Proposition~\ref{thm_smooth_det}, we first need to prove that its conditions are satisfied. Recall that $\gamma = \sqrt{\log d/(K^2\norm{W}_{\max}^2n_{\min})}$ and $n_{\min}\geq c_1 \left(s(p)^{\frac{2}{1-p}} \log d\right)/\Theta^2_{\min}\geq c_1 \log d$ for some $c_1$. It is easy to see that, for $c_1> 4$, we have $\gamma< {1}/({2K\norm{W}_{\max}})$ and hence, the first condition~\eqref{eq_1_smooth} is satisfied. To show the validity of the second condition~\eqref{eq_2_smooth}, we need the following intermediate lemma borrowed from~\cite{fattahi2021scalable}.

\begin{lem}[Theorem~3 of~\cite{fattahi2021scalable}]\label{lem_backward}
Under the conditions of Theorem~\ref{thm_smooth}, the following inequality holds with probability of at least $1-Kd^{-10}$
\begin{align*}
    \norm{\widetilde{F}^*(\widehat{\Sigma}_k) - \Theta^\star_{k}}_{\max}\leq\frac{128\kappa_1\kappa_3}{\kappa_2}\sqrt{\frac{\log d}{n_k}},\quad \text{for every $k$.}
\end{align*}
\end{lem}
\begin{proof}
The proof is a direct consequence of the proof of Theorem~3 in~\cite{fattahi2021scalable}. The details are omitted for brevity.
\end{proof}\vspace{2mm}

\noindent{\it Proof of Theorem~\ref{thm_smooth}.} Based on Lemma~\ref{lem_backward} and our choice of $\gamma$, one can write
\begin{align*}
    6\|\tilde{F}^*_{ij}\!-\!\Theta^\star_{ij}\|_{\infty}\!+\!6\gamma K\norm{W}_{\max} D
    \leq \left(\frac{768\kappa_1\kappa_3}{\kappa_2}+6D\right)\sqrt{\frac{\log d}{n_{\min}}} := \mu
\end{align*}
Now, it only remains to show that the defined $\mu$ indeed satisfies $\mu<3\Theta_{\min}/2$. This can be readily verified by our choice of $n_k$:
\begin{align*}
    n_k\geq c_2 \left(\frac{\kappa_1\kappa_3}{\kappa_2}+D\right)^2\frac{\log d}{\Theta_{\min}^2}
\end{align*}
for sufficiently large constant $c_2$. Therefore, the conditions of Proposition~\ref{thm_smooth_det} are satisfied, and as a result, we have $\supp(\widehat{\Theta}_{ij}) = \supp({\Theta}^\star_{ij})$ for every $(i,j)$ and 
$$
\norm{\widehat{\Theta}_k - \Theta^\star_k}_{\max}\leq 2\mu_3/3\lesssim\left(\frac{\kappa_1\kappa_3}{\kappa_2}+D\right)\sqrt{\frac{\log d}{n_{\min}}}
$$
which completes the proof of Theorem~\ref{thm_smooth}.$\hfill\square$

\section{Proof of Theorem~\ref{thm_sparse}}
Similar to the proof of Theorem~\ref{thm_smooth}, first we provide a deterministic guarantee on the estimation error.

\begin{prop}[Deterministic guarantee on sparsely-changing SV-MRFs]
	For a given $(i,j)$, suppose that the irrepresentability assumption is satisfied. Moreover, suppose that 
	\begin{align*}
	    \frac{8\kappa_{\mathrm{IC}}}{\alpha}\|\tilde{F}^*_{ij}\!-\!\Theta^\star_{ij}\|_{\infty} < \mu < \frac{\min\{2\Theta_{\min}, \Delta\Theta_{\min}\}}{4\sqrt{2\norm{W}_{\max} D_0}+5\sqrt{S_0}},
	\end{align*}
	Then, the following statements hold:\\
	\begin{itemize}
	    \item (Sparsistency) The solution $\widehat{\Theta}_{ij}$ is unique and satisfies $\supp(\widehat{\Theta}_{ij}) = \supp({\Theta}^\star_{ij})$ and $\supp(\widehat\Theta_{k;ij}-\widehat\Theta_{l;ij})=\supp(\Theta^\star_{k;ij}-\Theta^\star_{l;ij})$ for every $l>k$.
	    \item (Estimation error) We have 
	    \begin{align*}
    \left\|\widehat\Theta_{ij}-\Theta_{ij}^{\star}\right\|_{2}
    \leq 2\left(\sqrt{2\norm{W}_{\max}D_0}+1.25\sqrt{S_0}\right)\mu.
\end{align*}
	\end{itemize}
	\label{prop_sparse_det}
\end{prop}
The proof of the above proposition is provided in Appendix~\ref{proof_prop_sparse_det}. Based on the above proposition, we proceed with the proof of Theorem~\ref{thm_sparse}.\vspace{2mm}

\noindent{\it Proof of Theorem~\ref{thm_sparse}.}
In light of Lemma~\ref{lem_backward}, we have
\begin{align*}
    \norm{\widetilde{F}^*(\widehat{\Sigma}_k) - \Theta^\star_{k}}_{\max}\leq\frac{128\kappa_1\kappa_3}{\kappa_2}\sqrt{\frac{\log d}{n_k}},
\end{align*}
for every $k$ and with probability of at least $1-Kd^{-10}$. This implies that 
$$
\frac{8\kappa_{\mathrm{IC}}}{\alpha}\|\tilde{F}^*_{ij}\!-\!\Theta^\star_{ij}\|_{\infty}\leq \frac{1024\kappa_{\mathrm{IC}}\kappa_1\kappa_3}{\kappa_2\alpha}\sqrt{\frac{\log d}{n_{\min}}}:=\mu
$$
Now, in order to use Proposition~\ref{prop_sparse_det}, it suffices to have 
$$
\mu < \frac{\min\{2\Theta_{\min}, \Delta\Theta_{\min}\}}{4\sqrt{2\norm{W}_{\max} D_0}+5\sqrt{S_0}}
$$
which is guaranteed to hold once
$$
n_{\min}\geq c_3\left(\frac{\kappa_{\mathrm{IC}}\kappa_1\kappa_3}{\kappa_2\alpha}\right)\left(\frac{\norm{W}_{\max}D_0+S_0}{\min\{\Theta_{\min},\Delta\Theta_{\min}\}}\right)\cdot \log d
$$
for sufficiently large constant $c_3$. Therefore, Proposition~\ref{prop_sparse_det} holds and we have $\supp(\widehat{\Theta}_{k}) = \supp({\Theta}^\star_{k})$ for every $k$ and $\supp(\widehat\Theta_{k;ij}-\widehat\Theta_{l;ij})=\supp(\Theta^\star_{k;ij}-\Theta^\star_{l;ij})$ for every $i,j$ and $l>k$. Moreover, we have 
\begin{align*}
    \left\|\widehat\Theta_{ij}-\Theta_{ij}^{\star}\right\|_{2}
    &\leq 2\left(\sqrt{2\norm{W}_{\max}D_0}+1.25\sqrt{S_0}\right)\mu\\
    &\lesssim\left(\sqrt{\norm{W}_{\max}D_0}+\sqrt{S_0}\right) \frac{\kappa_{\mathrm{IC}}\kappa_1\kappa_3}{\kappa_2\alpha}\sqrt{\frac{\log d}{n_{\min}}}.
\end{align*}
This completes the proof.$\hfill\square$

\section{Proof of Proposition~\ref{prop_IC}}
Without loss of generality and to streamline the presentation, we assume that $W_{kl} = 1$ for every $k,l$. 
It is easy to observe that $B_{\mathcal{S}_B:}$ has full column rank, thus we have $\left(B_{\mathcal{S}_B:} \right)^{\dagger} = \left(B_{\mathcal{S}_B:}^\top B_{\mathcal{S}_B:}\right)^{-1}B_{\mathcal{S}_B:}^\top$. Let $(B_{\mathcal{S}_B:}^\top B_{\mathcal{S}_B:})^{-1}={\mathrm{det}(B_{\mathcal{S}_B:}^\top B_{\mathcal{S}_B:})}^{-1} F$, where $F$ is the adjugate of $B_{\mathcal{S}_B:}^\top B_{\mathcal{S}_B:}$ defined as $F_{ij} = (-1)^{i+j}M_{ij}$ and $M_{ij}$ is the minor of $B_{\mathcal{S}_B:}^\top B_{\mathcal{S}_B:}$ formed by deleting its $i$-th row and $j$-th column. We split $\mathcal{S}_B$ into two sets $\mathcal{S}_1$ and $\mathcal{S}_2$, where $\mathcal{S}_1$ represents the support set of $GA\Theta_{ij}^\star$ and $\mathcal{S}_2 = \mathcal{S}_B \backslash \mathcal{S}_1$. We also split $\mathcal{S}_B^c$ into two sets $\mathcal{S}_1^c$ and $\mathcal{S}_2^c$ in a similar way. Our first goal is to show that IC holds with $\alpha = \mu/\gamma$ which satisfies $1/2\leq \alpha\leq 1$. To this end, we show that 
\begin{align}
        &B_{\mathcal{S}^c_1}B_{\mathcal{S}_B:}^{\dagger} \operatorname{sign}\left((B\Theta_{ij}^{\star})_{\mathcal{S}_B:}\right) =0,\label{eq_S1}\\
        &\left\|B_{\mathcal{S}^c_2}B_{\mathcal{S}_B:}^{\dagger} \operatorname{sign}\left((B\Theta_{ij}^{\star})_{\mathcal{S}_B:}\right)\right\|_{\infty} \leq 1-\mu/\gamma\label{eq_S2}
\end{align}
To prove~\eqref{eq_S1}, it suffices to show that 
\begin{align*}
    \left[(B_{\mathcal{S}_B:}^\top B_{\mathcal{S}_B:})^{-1} B_{\mathcal{S}_B:}^\top \mathrm{sign}(B_{\mathcal{S}_B:} \Theta_{ij}^\star)\right]_u 
    = \left[(B_{\mathcal{S}_B:}^\top B_{\mathcal{S}_B:})^{-1} B_{\mathcal{S}_B:}^\top \mathrm{sign}(B_{\mathcal{S}_B:} \Theta_{ij}^\star)\right]_v, 
\end{align*}
for every $u\not=v$ such that $\Theta_{u;ij}^\star = \Theta_{v;ij}^\star$. First, we write
\begin{equation}\nonumber
\begin{aligned}
    {} \left[(B_{\mathcal{S}_B:}^\top B_{\mathcal{S}_B:})^{-1} B_{\mathcal{S}_B:}^\top \mathrm{sign}(B_{\mathcal{S}_B:} \Theta_{ij}^\star)\right]_u
    =& \left[\left(B_{\mathcal{S}_B:}^\top B_{\mathcal{S}_B:}\right)^{-1}\right]_{u:}B_{\mathcal{S}_B:}^\top\operatorname{sign}\left(B_{\mathcal{S}_B:}\Theta_{ij}^{\star}\right)\\
    =& \left[\left(B_{\mathcal{S}_B:}^\top B_{\mathcal{S}_B:}\right)^{-1}\right]_{u:}\left[B_{\mathcal{S}_1}^\top \ \  B_{\mathcal{S}_2}^\top \right]\operatorname{sign}\left(B_{\mathcal{S}_B:}\Theta_{ij}^{\star}\right)\\
    =& \frac{\frac{\gamma}{\mu} \sum\limits_{k<l}(F_{uk}-F_{ul}) \operatorname{sign}(\Theta_{k;ij}^\star - \Theta_{l;ij}^\star)+\sum\limits_{k}F_{uk}\operatorname{sign}(\Theta_{k;ij}^\star)}{\mathrm{det}(B_{\mathcal{S}_B:}^\top B_{\mathcal{S}_B:})}
\end{aligned}
\end{equation}
where the last equality is due to the definition of $F$. On the other hand, one can verify that
  \begin{align}
    (B_{\mathcal{S}_B:}^\top B_{\mathcal{S}_B:})_{kl}
    =
    \begin{cases}
      -\frac{\gamma^2}{\mu^2}\mathbb{I}{\{\pi(k,l)\in \mathcal{S}_1\}}, & k\not = l \\
      \mathbb{I}{\{\Theta_{k;ij}^\star \not= 0\}}-\sum_{u\not=k}( B_{\mathcal{S}_B:}^\top B_{\mathcal{S}_B:})_{ku}, & k=l
    \end{cases}\label{eq:BSTBS}
  \end{align}
where $\mathbb{I}(\cdot)$ is the indicator function. Before proceeding, we need the following intermediate lemma.
\begin{lem}\label{lem:FF}
    The following statements hold:
    \begin{itemize}
        \item $F_{uu} = F_{vv}$ for every $u, v$ such that $\Theta_{u;ij}^\star = \Theta_{v;ij}^\star$.
	\item $F_{uk} = F_{vk}$ for every $u,v,k$ such that $\Theta_{u;ij}^\star = \Theta_{v;ij}^\star$, $k \not = u$, and $k \not = v$.
    \end{itemize}
\end{lem}

The proof of the above lemma can be found in Appendix~\ref{app_lem:FF}.
For $u, v$ such that $\Theta_{u;ij}^\star = \Theta_{v;ij}^\star$, one can write
\begin{equation}\nonumber
\begin{aligned}
    {} \left[(B_{\mathcal{S}_B:}^\top B_{\mathcal{S}_B:})^{-1} B_{\mathcal{S}_B:}^\top \mathrm{sign}(B_{\mathcal{S}_B:} \Theta_{ij}^\star)\right]_u 
    -& \left[(B_{\mathcal{S}_B:}^\top B_{\mathcal{S}_B:})^{-1} B_{\mathcal{S}_B:}^\top \mathrm{sign}(B_{\mathcal{S}_B:} \Theta_{ij}^\star)\right]_v\\
    =& \frac{1}{\mathrm{det}(B_{\mathcal{S}_B:}^\top B_{\mathcal{S}_B:})}\bigg(\frac{\gamma}{\mu} \sum\limits_{k<l}\left((F_{uk}-F_{vk})-(F_{ul}-F_{vl})\right)\\
    &\mathrm{sign}(\Theta_{k;ij}^\star - \Theta_{l;ij}^\star)+\sum\limits_{k}(F_{uk}-F_{vk})\mathrm{sign}(\Theta_{k;ij}^\star)\bigg)
\end{aligned}
\end{equation}
Therefore, it suffices to prove that $\sum\limits_{k<l}\left((F_{uk}-F_{vk})-(F_{ul}-F_{vl})\right)\mathrm{sign}(\Theta_{k;ij}^\star - \Theta_{l;ij}^\star) = 0$ and $\sum\limits_{k}(F_{uk}-F_{vk})\mathrm{sign}(\Theta_{k;ij}^\star) = 0$. Invoking Lemma \ref{lem:FF}, one can write
\begin{equation*}
\begin{aligned}
     {} \sum\limits_{k<l}(F_{uk}-F_{vk})\operatorname{sign}(\Theta_{k;ij}^\star - \Theta_{l;ij}^\star)
     = & \sum\limits_{u<l}(F_{uu}-F_{vu})\operatorname{sign}(\Theta_{u;ij}^\star - \Theta_{l;ij}^\star)\\
     &+  \sum\limits_{v<l}(F_{uv}-F_{vv})\operatorname{sign}(\Theta_{v;ij}^\star - \Theta_{l;ij}^\star)\\
     = & (F_{uu}-F_{uv})\bigg(\sum\limits_{u<l}\operatorname{sign}(\Theta_{u;ij}^\star\\
     &\hspace{2.2cm}- \Theta_{l;ij}^\star) - \sum\limits_{v<l}\operatorname{sign}(\Theta_{u;ij}^\star - \Theta_{l;ij}^\star) \bigg)\\
     = & (F_{uu}-F_{uv})\sum\limits_{l=u+1}^{v}\operatorname{sign}(\Theta_{u;ij}^\star - \Theta_{l;ij}^\star).
\end{aligned}
\end{equation*}
Thus, we have
\begin{equation*}
\begin{aligned}
     {} \sum\limits_{k<l}\left((F_{uk}-F_{vk})-(F_{ul}-F_{vl})\right) \operatorname{sign}(\Theta_{k;ij}^\star - \Theta_{l;ij}^\star)
     = & (F_{uu}-F_{uv})(\sum\limits_{l=u+1}^{v}\operatorname{sign}(\Theta_{u;ij}^\star - \Theta_{l;ij}^\star) \\
     &+ \sum\limits_{k=u+1}^{v}\operatorname{sign}(\Theta_{k;ij}^\star - \Theta_{u;ij}^\star)) = 0.
\end{aligned}
\end{equation*}

On the other hand, one can write
\begin{equation*}
\begin{aligned}
     {} \sum\limits_{k}(F_{uk}-F_{vk})\operatorname{sign}(\Theta_{k;ij}^\star)
     = & (F_{uu}-F_{vu})\operatorname{sign}(\Theta_{u;ij}^\star) + (F_{uv}-F_{vv})\operatorname{sign}(\Theta_{v;ij}^\star)\\
     = & (F_{uu}-F_{uv})\operatorname{sign}(\Theta_{u;ij}^\star) - (F_{uu}-F_{uv})\operatorname{sign}(\Theta_{v;ij}^\star)\\
     = & 0,
\end{aligned}
\end{equation*}
which completes the proof of~\eqref{eq_S1}.

 Now, to prove~\eqref{eq_S2}, we derive the explicit form of $\left[\left(B_{\mathcal{S}_B:}^\top  B_{\mathcal{S}_B:}\right)^{-1} B_{\mathcal{S}_B:}^\top \operatorname{sign}\left(B_{\mathcal{S}_B:}\Theta_{ij}^{\star}\right)\right]_u$, for any $u$ such that $\Theta_{u;ij}^\star = 0$. Recall that the determinant of a matrix remains unchanged after adding multiples of a column to another column.  Adding all the other columns to the column $u$ of $B_{\mathcal{S}_B:}^\top B_{\mathcal{S}_B:}$ changes this column to $[\mathbb{I}_{\{\Theta_{1;ij}^\star \not= 0\}}, \mathbb{I}_{\{\Theta_{2;ij}^\star \not= 0\}}, ..., \mathbb{I}_{\{\Theta_{K;ij}^\star \not= 0\}}]^\top$.
Therefore, for any $\frac{\gamma}{\mu} > 0$, we have
\begin{equation}\nonumber
\begin{aligned}
    &\sum_{k\in S} -\frac{\gamma^2}{\mu^2}F_{ku} + |S|\frac{\gamma^2}{\mu^2}F_{uu} = \sum_{k\in S} F_{ku} = \det(B_{\mathcal{S}_B:}^\top B_{\mathcal{S}_B:})\\
&\implies F_{uu} = \frac{(1+\frac{\gamma^2}{\mu^2})\det(B_{\mathcal{S}_B:}^\top B_{\mathcal{S}_B:})}{|S|\frac{\gamma^2}{\mu^2}}
\end{aligned}
\end{equation}
On the other hand, for every $v\in S_2^c \backslash u$, one can write
\begin{align*}
    \left[B_{\mathcal{S}_B:}^\top B_{\mathcal{S}_B:}\right]_{v:} \left[(B_{\mathcal{S}_B:}^\top B_{\mathcal{S}_B:})^{-1}\right]_{:u}
    =
    \frac{\sum_{k\in S} -\frac{\gamma^2}{\mu^2}F_{ku} + |S|\frac{\gamma^2}{\mu^2}F_{vu}}{ \det(B_{\mathcal{S}_B:}^\top B_{\mathcal{S}_B:})}=0,
\end{align*}
Recalling that $\sum_{k\in S}F_{ku} = \det(B_{\mathcal{S}_B:}^\top B_{\mathcal{S}_B:})$, the above equality implies $F_{uv}=\frac{\det(B_{\mathcal{S}_B:}^\top B_{\mathcal{S}_B:})}{|S|}$.
Similarly, for $k\in S_2$, one can write

\begin{align}
    \left[B_{\mathcal{S}_B:}^\top B_{\mathcal{S}_B:}\right]_{u:} \left[(B_{\mathcal{S}_B:}^\top B_{\mathcal{S}_B:})^{-1}\right]_{: k} 
    =  
    \frac{\sum_{l\in S} -\frac{\gamma^2}{\mu^2}F_{lk} + |S|\frac{\gamma^2}{\mu^2}F_{uk}}{ \det(B_{\mathcal{S}_B:}^\top B_{\mathcal{S}_B:})}=0.
\end{align}
which again implies $F_{uk}=\frac{\det(B_{\mathcal{S}_B:}^\top B_{\mathcal{S}_B:})}{|S|}$.
Therefore, for any $u\in S^c_2$, we have 
  \begin{equation}
    F_{uk}=
    \begin{cases}
      \frac{\det(B_{\mathcal{S}_B:}^\top B_{\mathcal{S}_B:})}{|S|}, & k\not = u \\
      \frac{(1+\frac{\gamma^2}{\mu^2})\det(B_{\mathcal{S}_B:}^\top B_{\mathcal{S}_B:})}{|S|\frac{\gamma^2}{\mu^2}}, & k=u
    \end{cases}\label{eq:F_uk}
  \end{equation}
Now, based on the above formula, we can compute the explicit form of $\left[\left(B_{\mathcal{S}_B:}^\top  B_{\mathcal{S}_B:}\right)^{-1} B_{\mathcal{S}_B:}^\top \operatorname{sign}\left(B_{\mathcal{S}_B:}\Theta_{ij}^{\star}\right)\right]_u$. First, note that
\begin{equation*}
    \begin{aligned}
    {} \sum\limits_{k<l}(F_{uk}-F_{ul}) \operatorname{sign}(\Theta_{k;ij}^\star - \Theta_{l;ij}^\star)
    =&\sum\limits_{k<u}(F_{uk}-F_{uu}) \mathrm{sign}(\Theta_{k;ij}^\star - \Theta_{u;ij}^\star)\\
    &+ \sum\limits_{u<l}(F_{uu}-F_{ul}) \mathrm{sign}(\Theta_{u;ij}^\star - \Theta_{l;ij}^\star)\\
    =&\sum\limits_{k<u}(F_{uk}-F_{uu}) \mathrm{sign}(\Theta_{k;ij}^\star) + \sum\limits_{u<l}(F_{ul}-F_{uu}) \mathrm{sign}(\Theta_{l;ij}^\star)\\
    =& \sum\limits_{k\not = u}(F_{uk}-F_{uu}) \mathrm{sign}(\Theta_{k;ij}^\star)\\
    =& \sum\limits_{k \in S}(F_{uk}-F_{uu}) \mathrm{sign}(\Theta_{k;ij}^\star)
    \end{aligned}
\end{equation*}
Therefore
\begin{equation}
    \begin{aligned}
    {}  \left[\left(B_{\mathcal{S}_B:}^\top B_{\mathcal{S}_B:}\right)^{-1}B_{\mathcal{S}_B:}^\top\operatorname{sign}\left(B_{\mathcal{S}_B:}\Theta_{ij}^{\star}\right)\right]_u 
    &= \frac{\frac{\gamma}{\mu} \sum\limits_{k \in S}(F_{uk}-F_{uu}) \mathrm{sign}(\Theta_{k;ij}^\star) +\sum\limits_{k}F_{uk}\operatorname{sign}(\Theta_{k;ij}^\star)}{\mathrm{det}(B_{\mathcal{S}_B:}^\top B_{\mathcal{S}_B:})}\\
    &=\frac{(\frac{\gamma}{\mu}+1)\sum\limits_{k\in S}F_{uk}\mathrm{sign}(\Theta_{k;ij}^\star) - \frac{\gamma}{\mu}\sum\limits_{k\in S}F_{uu}\mathrm{sign}(\Theta_{k;ij}^\star)}{\det(B_{\mathcal{S}_B:}^\top B_{\mathcal{S}_B:})}\\
    &=\frac{(\frac{\gamma}{\mu}+1)\sum\limits_{k\in S}\mathrm{sign}(\Theta_{k;ij}^\star)}{|S|} - \frac{\gamma}{\mu}\sum\limits_{k\in S}\mathrm{sign}(\Theta_{k;ij}^\star)\frac{1+\frac{\gamma^2}{\mu^2}}{|S|\frac{\gamma^2}{\mu^2}}\\
        &=\frac{\sum\limits_{k\in S}\mathrm{sign}(\Theta_{k;ij}^\star)}{|S|}\left(1-\frac{\mu}{\gamma}\right) \leq 1-\frac{\mu}{\gamma}
    \end{aligned}
    \label{eq:irrepresentability_S2c}
\end{equation}
where in the third equality, we used the explicit form of $F_{uk}$ in~\eqref{eq:F_uk}. Therefore, we have $\left\|B_{\mathcal{S}^c:}\left(B_{\mathcal{S}_B:}^\top B_{\mathcal{S}_B:}\right)^{-1}B_{\mathcal{S}_B:}^\top\mathrm{sign}\left(B_{\mathcal{S}_B:}\Theta_{ij}^{\star}\right)\right\|_{\infty} \leq |1-\mu/\gamma|$, which implies that $\alpha = \mu/\gamma$. 

Next, we provide an upper and lower bound for $\kappa_{\mathrm{IC}}$. Recall that $\kappa_{\mathrm{IC}} := \left\|B_{\mathcal{S}^c:}\left(B_{\mathcal{S}_B:}^\top B_{\mathcal{S}_B:}\right)^{-1}B_{\mathcal{S}_B:}^\top\right\|_{\infty}+1$. Hence, we trivially have $\kappa_{\mathrm{IC}}\geq 1$. Therefore, it suffices to show that $\left\|B_{\mathcal{S}^c:}\left(B_{\mathcal{S}_B:}^\top B_{\mathcal{S}_B:}\right)^{-1}B_{\mathcal{S}_B:}^\top\right\|_{\infty}\leq 4$. To this goal, we show that $\left\|B_{u:}\left(B_{\mathcal{S}_B:}^\top B_{\mathcal{S}_B:}\right)^{-1}B_{\mathcal{S}_B:}^\top\right\|_1\leq 4$, for every $u\in\mathcal{S}_B^c$. We consider three cases:
\vspace{2mm}

\noindent{\it Case 1:} Suppose that $u = \pi(k, l)$ for some $(k,l)$ such that $\Theta_{k;ij}^\star = \Theta_{l;ij}^\star = 0$.
One can write
    \begin{equation}
    \begin{aligned}
        {} B_{u:}\left(B_{\mathcal{S}_B:}^\top B_{\mathcal{S}_B:}\right)^{-1}B_{\mathcal{S}_B:}^\top
        = & \frac{\gamma}{\mu \det(B_{\mathcal{S}_B:}^\top B_{\mathcal{S}_B:})}\left[F_{k1}-F_{l1},...,F_{kK}-F_{lK} \right]B_{\mathcal{S}_B:}^\top\\
        = & \frac{\gamma}{\mu \det(B_{\mathcal{S}_B:}^\top B_{\mathcal{S}_B:})}\\
        &\cdot \!\left[...,0, F_{kk}\!-\!F_{lk},0,...,0,F_{kl}\!-\!F_{ll},0,... \right]\!B_{\mathcal{S}_B:}^\top\\
        = & \frac{\gamma}{\mu}\left[..., 0,\frac{\mu^2}{\gamma^2|S|},0,...,0,\frac{-\mu^2}{\gamma^2|S|},0,... \right]B_{\mathcal{S}_B:}^\top
    \end{aligned}
    \end{equation}
Thus $\| B_{u:}\left(B_{\mathcal{S}_B:}^\top B_{\mathcal{S}_B:}\right)^{-1}B_{\mathcal{S}_B:}^\top\|_{1} = 2\frac{\gamma}{\mu}\sum\limits_{u\in S} \frac{\mu^2}{\gamma^2|S|} \leq  4$.
\vspace{2mm}

\noindent{\it Case 2:} Suppose that $u = \pi(k, l)$ for some $(k,l)$ such that $\Theta_{k;ij}^\star = \Theta_{l;ij}^\star \not= 0$.
Let $\bar{S}= \{r| \Theta_{r;ij}^\star \not= \Theta_{k;ij}^\star, r=1,...,K\}$. It is easy to verify that column $k$ and $l$ of $B_{\mathcal{S}_B:}^\top B_{\mathcal{S}_B:}$ are the same except for the $k$-th and $l$-th element. Subtracting column $k$ of $B_{\mathcal{S}^\top B:}$ from column $l$, one can write 
    \begin{equation}
        (|\bar{S}|\frac{\gamma^2}{\mu^2}+1)F_{kk} - (|\bar{S}|\frac{\gamma^2}{\mu^2}+1)F_{kl} = \det(B_{\mathcal{S}_B:}^\top B_{\mathcal{S}_B:}).
    \end{equation}
    Therefore, 
    \begin{equation}
        F_{kk} - F_{kl} = \frac{\det(B_{\mathcal{S}_B:}^\top B_{\mathcal{S}_B:})}{|\bar{S}|\frac{\gamma^2}{\mu^2}+1}
        \label{eq:Fkk-Fkl}
    \end{equation}
    Combining \eqref{eq:Fkk-Fkl} and Lemma \ref{lem:FF}, one can write
    \begin{equation}
        \begin{aligned}
            {}  B_{u:}\left(B_{\mathcal{S}_B:}^\top B_{\mathcal{S}_B:}\right)^{-1}B_{\mathcal{S}_B:}^\top
            & = \frac{\gamma}{\mu \det(B_{\mathcal{S}_B:}^\top B_{\mathcal{S}_B:})}\left[F_{k1}-F_{l1},...,F_{kK}-F_{lK} \right]B_{\mathcal{S}_B:}^\top\\
            & = \frac{\gamma}{\mu}\left[...,0, \frac{1}{|\bar{S}|\frac{\gamma^2}{\mu^2}+1},0,,...,0,- \frac{1}{|\bar{S}|\frac{\gamma^2}{\mu^2}+1}, 0,... \right]B_{\mathcal{S}_B:}^\top
        \end{aligned}
    \end{equation}
    Therefore, again we have $\| B_{u:}\left(B_{\mathcal{S}_B:}^\top B_{\mathcal{S}_B:}\right)^{-1}B_{\mathcal{S}_B:}^\top\|_{1} \leq  4$.
    \vspace{2mm}
    
    \noindent{\it Case 3:} Suppose that $u = \frac{K(K-1)}{2}+k$ for some $k$ and $\Theta_{k;ij}^\star = 0$. One can write
    \begin{equation}
    \begin{aligned}
        {}  B_{u:}\left(B_{\mathcal{S}_B:}^\top B_{\mathcal{S}_B:}\right)^{-1}B_{\mathcal{S}_B:}^\top
        = & \frac{1}{\det(B_{\mathcal{S}_B:}^\top B_{\mathcal{S}_B:})}\left[F_{k1}, F_{k2},...,F_{kK} \right] B_{\mathcal{S}_B:}^\top\\
        = & \frac{\gamma}{\mu}\left[..., \frac{1}{|S|},\frac{1+\gamma^2/\mu^2}{|S|(\gamma^2/\mu^2)}, \frac{1}{|S|},... \right] B_{\mathcal{S}_B:}^\top
    \end{aligned}
    \end{equation}
    Thus $\lonenorm{ B_{u:}\left(B_{\mathcal{S}_B:}^\top B_{\mathcal{S}_B:}\right)^{-1}B_{\mathcal{S}_B:}^\top} = 1 + \frac{\mu}{\gamma} \leq 3$.
The completes the proof.
$\hfill\square$

\section{Additional Proofs}
\subsection{Proof of Proposition~\ref{thm_smooth_det}}\label{proof_cor_smooth_det}
For the purpose of proof, we rewrite the optimization problem~\eqref{problem:lq_ij2} with $q=2$ as an instance of the Lasso problem~\cite{wainwright_2019}. 
\begin{equation}
	\begin{aligned}
		\ltwonorm{\Theta_{ij} - \tilde{F}^*_{ij}}^2 + \gamma \ltwonorm{GA\Theta_{ij}}^2
		&=(\Theta_{ij}-\tilde{F}^*_{ij})^{\top}(\Theta_{ij}-\tilde{F}^*_{ij}) + \gamma\Theta_{ij}^{\top} A^{\top}G^{\top}GA\Theta_{ij}\\
		&=\Theta_{ij}^{\top}( I+\gamma A^{\top}G^{\top}GA)\Theta_{ij} - 2\Theta_{ij}^{\top}\tilde{F}^*_{ij} + (\tilde{F}^*_{ij})^{\top}\tilde{F}^*_{ij}
	\end{aligned}  
\end{equation}
Since $I+\gamma A^{\top}G^{\top}GA$ is a strictly diagonally dominant symmetric matrix, it has a unique Cholesky decomposition $I+\gamma A^{\top}G^{\top}GA = C^\top C$. Therefore, one can write
\begin{equation}
    \begin{aligned}
     \ltwonorm{\Theta_{ij} - \tilde{F}^*_{ij}}^2 + \gamma \ltwonorm{GA\Theta_{ij}}^2 &=\ltwonorm{(C^\top)^{-1}\tilde{F}^*_{ij} - C\Theta_{ij}}^2
    -(\tilde{F}^*_{ij})^{\top} C^{-1}(C^\top)^{-1}\tilde{F}^*_{ij} +(\tilde{F}^*_{ij})^{\top}\tilde{F}^*_{ij}
    \end{aligned}
\end{equation}
Therefore, problem (\ref{problem:lq_ij}) is equivalent to:
\begin{equation}
	\centering
	\begin{aligned} \min \quad& \ltwonorm{y - C\Theta_{ij}}^2 + \mu \lonenorm{\Theta_{ij}}, \end{aligned}
	\label{lasso_problem_2}
\end{equation}
where $y = (C^\top)^{-1}\tilde{F}^*_{ij}$. Note that~\eqref{lasso_problem_2} is an instance of Lasso with the observation model $y = C\Theta_{ij}^\star + w$, where $w = (C^\top)^{-1}\tilde{F}^*_{ij} - C\Theta^\star_{ij}$ is the noise vector. 

	The above equivalence will allow us to invoke the exact recovery guarantee of Lasso in our setting. For any fixed $(i,j)$, let $S\subset \{1,2,...,K\}$ with $|S| = s$ be the support of $\Theta_{ij}^\star$ satisfying $\Theta_{k;ij}^\star \not = 0$ for all $k\in S$. Moreover, let $S^c = \{1,2,...,K\}\backslash S$. Without loss of generality and to streamline our presentation, we assume that $S = \{1,2,\dots,s\}$. At the crux of the exact recovery guarantees based on Lasso lies the classical notion of \textit{mutual incoherency}~\cite{zhao2006model, candes2007sparsity, wainwright_2019}.
	
	\begin{assumption}[Mutual Incoherency]
	    There exists some $\alpha \in[0,1)$ such that
	\begin{equation*}
		\max _{j \in S^{c}}\left\|\left(C_{S}^{\mathrm{T}} C_{S}\right)^{-1} C_{S}^{\mathrm{T}} C_{j}\right\|_{1} \leq \alpha
	\end{equation*}
	\end{assumption}
	Mutual incoherency entails that the effect of the columns of $C$ corresponding to ``unimportant'' (zero) elements of $\Theta^\star_{ij}$ on the remaining columns is small. Although mutual incoherency cannot be guaranteed for general choices of $C$ and $S$, our next lemma shows that it is indeed satisfied for our problem, provided that $\gamma$ is sufficiently small.
	\begin{lem}
		For $C^\top C = I+\gamma A^{\top}G^{\top}GA$ and any $S\subseteq\{1,2,\dots, K\}$, the mutual incoherency holds with $\alpha=1/2$, provided that $\gamma < 1/(2K\norm{W}_{\max})$.
		\label{lem:multual_incoherence}
	\end{lem}
	\begin{proof}
	It is easy to verify that
	\begin{align}
	    A^{\top}G^{\top}G A=\begin{bmatrix}
		\displaystyle\sum_{i\not=1} W_{1,i} & -W_{1,2} &\cdots &-W_{1,K}\\
		-W_{2,1} & \displaystyle\sum_{i\not=2} W_{2,i} &\cdots &-W_{2,K}\\
		&  & & \cdots   &\\
		-W_{K,1} & -W_{K,2}  &\cdots & \displaystyle\sum_{i\not=K} W_{K,i}
	\end{bmatrix}\label{eq_AGGA}
	\end{align}
	Since $S \cap S^c = \emptyset$, the elements of $C^\top_SC_{S^c}$ can only include the off-diagonal elements of $C^\top C$, which in turn correspond to the off-diagonal entries of $A^{\top}G^{\top}G A$ defined as~\eqref{eq_AGGA}. Therefore, we have
	
	\begin{equation}
		\linfinity{C^\top_SC_j} \leq \gamma \norm{W}_{\max}, \quad j \in S^c.
		\label{ieq:inf_off}
	\end{equation}
	On the other hand, since $C^\top_SC_S$ is strictly diagonally dominant, its inverse matrix satisfies
	
	\begin{equation}
		\begin{aligned}
			\quad \linfinity{(C^\top_SC_S)^{-1}}
			&\leq \frac{1}{\min_i \left\{\abs{(C^\top_SC_S)_{i,i}} - \sum _{j\not = i}\abs{(C^\top_SC_S)_{i,j}}\right\}}\\
			& = \frac{1}{\min_i\left\{1+\gamma\sum_{j\not=i}\abs{W_{i,j}}-\gamma\sum_{j\in S\backslash i}\abs{W_{i,j}} \right\}}\\
			& = \frac{1}{\min_i\left\{1+\gamma\sum_{j\in S^c\backslash i}\abs{W_{i,j}} \right\}}
			\leq 1
		\end{aligned}
		\label{ieq:inv_CC}
	\end{equation}
	where the first inequality is due to \cite[Theorem 1]{varah1975lower}. Combined with (\ref{ieq:inf_off}), we have
	\begin{equation}
		\begin{aligned}
			\lonenorm{(C_{S}^{\top} C_{S})^{-1} C_{S}^{\top} C_{j}} &\leq K\linfinity{(C_S^{\top} C_S)^{-1}}\linfinity{C_S^{\top}C_j}\\
			& \leq K\gamma\norm{W}_{\max}\linfinity{(C_S^{\top} C_S)^{-1}}\\
			& \leq K\gamma \norm{W}_{\max} < 1/2,
		\end{aligned}
		\label{ieq:3}
	\end{equation}
	which completes the proof.
	
\end{proof}
	Given the mutual incoherency condition, we follow the so-called \textit{primal-dual witness (PDW)} approach introduced by~\citet{wainwright_2019} to prove Proposition~\ref{thm_smooth_det}. To this goal, first we delineate the optimality conditions for~\eqref{lasso_problem_2}. Given a convex function $f\colon \RR^{K} \rightarrow \RR$, we say that $z \in \RR^{K}$ is a subgradient of $f$ at $\Theta_{ij}$, denoted by $z \in \partial f(\Theta_{ij})$, if we have 
	\[
	f(\Theta_{ij}+\delta) \geq f(\Theta_{ij}) + \langle z,\delta \rangle \text{ for all } \delta \in \RR^{K}.
	\]
	When $f(\Theta_{ij})=\lonenorm{\Theta_{ij}}$, it can be seen that $z \in \partial \lonenorm{\Theta_{ij}}$ if and only if $z_k=\mathrm{sign}(\Theta_{k;ij})$ for all $k=1,2,...,K$. For (\ref{lasso_problem_2}), we say that a pair $(\widehat\Theta_{ij}, \hat{z})$ is \emph{primal-dual optimal} if 
	\begin{equation}
		\hat{z} \in \partial \lonenorm{\Theta_{ij}}, \quad 2C^\top(C\widehat\Theta_{ij}-y) + \mu \hat{z}  = 0
		\label{eq:zero-subg}
	\end{equation}
	Evidently, if $(\widehat\Theta_{ij}, \hat{z})$ is primal-dual optimal, then $\widehat\Theta_{ij}$ is the minimizer of~\eqref{lasso_problem_2}. Given this optimality condition, the {Primal-dual witness} is constructed as follows.
	
	\vspace{2mm}
	\hrule
	\hrule
	\vspace{2mm}
	\noindent\textbf{PDW construction: } PDW construction has the following steps:
	
	\begin{enumerate}
		\item $\operatorname{Set} \widehat\Theta_{S^{c};ij}=0$
		\item Determine $\left(\widehat\Theta_{S;ij}, \widehat{z}_{S}\right) \in \mathbb{R}^{s} \times \mathbb{R}^{s}$ by solving the oracle subproblem
		\begin{equation}
			\widehat\Theta_{S;ij} \in \arg \min _{\Theta_{S;ij} \in \mathbb{R}^{s}}\{\underbrace{\ltwonorm{y-C_S\Theta_{S;ij}}^2}_{=: f\left(\Theta_{S;ij}\right)}+\mu \lonenorm{\Theta_{S;ij}},
			\label{eq:subproblem}
		\end{equation}
		then choose $\widehat{z}_{S} \in \partial\left\|\widehat\Theta_{S;ij}\right\|_{1}$ such that $\left.\nabla f\left(\Theta_{S;ij}\right)\right|_{\Theta_{S;ij}=\widehat\Theta_{S;ij}}+\mu \widehat{z}_{S}=0$
		\item Solve for $\widehat{z}_{S^{c}} \in \mathbb{R}^{K-s}$ via the zero-subgradient equation (\ref{eq:zero-subg}), and check whether or not the strict dual feasibility condition $\left\|\widehat{z}_{S^c}\right\|_{\infty}<1$ holds.
	\end{enumerate}
	\vspace{2mm}
	\hrule
	\hrule
	\vspace{2mm}
	
	Note that the vector $\widehat\Theta_{S^{c};ij}$ is determined in Step 1, whereas the remaining three subvectors $\widehat\Theta_{S;ij}$, $\widehat{z}_S$, and $\widehat{z}_{S^c}$ are determined in Steps 2 and 3. By construction, $\widehat\Theta_{S^{c};ij}=0$ and the subvectors $\widehat\Theta_{S;ij}, \widehat{z}_{S}$ and $\widehat{z}_{S^c}$ satisfy the zero-subgradient condition (\ref{eq:zero-subg}). Therefore, we have
	
	\begin{equation}
	\begin{aligned}
		&\left[\begin{array}{cc}
			C_{S}^\top C_{S} & C_{S}^\top C_{S^{c}} \\
			C_{S^{c}}^\top C_{S} & C_{S^{c}}^\top C_{S^{c}}
		\end{array}\right]\left[\begin{array}{c}
			\widehat\Theta_{S;ij} - \Theta_{S;ij}^\star \\
			0
		\end{array}\right]
		+ \frac{\mu}{2}\left[\begin{array}{c}
			\widehat{z}_{S} \\
			\widehat{z}_{S^{c}}
		\end{array}\right]=\left[\begin{array}{c}
			C^\top_S w \\
			C^\top_{S^c} w
		\end{array}\right]\\
	\end{aligned}
		\label{eq:matrix_zero}
	\end{equation}
	We say that ``PDW construction succeeds'' if the vector $\widehat{z}_{S^c}$ constructed in step 3 satisfies the strict dual feasibility condition. The following result shows that the success of PDW construction implies that $(\widehat\Theta_{S;ij}, 0) \in \RR^{K}$ is the unique solution of the Lasso problem.
	
	\begin{lem}[Lemma 7.23 of \cite{wainwright_2019}]
		The success of the PDW construction implies that the vector $(\widehat\Theta_{S;ij}, 0) \in \RR^{K}$ is the unique optimal solution of the Lasso problem (\ref{lasso_problem_2}).
		\label{lem:successPDW}
	\end{lem}
	
	\noindent{\it Proof of Proposition~\ref{thm_smooth_det}.}
	To apply Lemma~\ref{lem:successPDW}, it suffices to show that the vector $\widehat{z}_{S^c}\in \RR^{K-s}$ constructed in Step 3 satisfies the strict dual feasibility condition. Using the zero-subgradient condition (\ref{eq:matrix_zero}), we have
	
	\begin{equation}
		\widehat{z}_{S^c} = -\frac{2}{\mu}C^\top_{S^c}C_S(\widehat\Theta_{S;ij} - \Theta_{S;ij}^\star) + \frac{2}{\mu}C^\top_{S^c}w.
		\label{eq:solve_z}
	\end{equation}
	On the other hand, (\ref{eq:matrix_zero}) implies that
	\begin{equation}
		\widehat\Theta_{S;ij} - \Theta_{S;ij}^\star= (C^\top_SC_S)^{-1}C^\top_Sw - \frac{\mu}{2}(C^\top_SC_S)^{-1}\widehat{z}_S
		\label{eq:error}
	\end{equation}
	Substituting this expression back into (\ref{eq:solve_z}) yields
	\begin{equation}
		\widehat{z}_{S^c} = C^\top_{S^c}C_S(C^\top_SC_S)^{-1}\widehat{z}_{S} + C^\top_{S^c}[I - C_S(C^\top_SC_S)^{-1}C^\top_S]\frac{2w}{\mu}
	\end{equation}
	Due to Lemma~\ref{lem:multual_incoherence}, we have 
	$\linfinity{C^\top_{S^c}C_S(C^\top_SC_S)^{-1}\widehat{z}_{S}} \leq 1/2$. Therefore, if the following inequalities are satisfied
	\begin{align}
		&\left\|{C^\top_{S^c}[I - C_S(C^\top_SC_S)^{-1}C^\top_S]\frac{2w}{\mu}}\right\|_{\infty}\leq \frac{1}{2}\nonumber\\
		\implies & 4\linfinity{C^\top_{S^c}w-C^\top_{S^c}C_{S}\left(C_{S}^{\top} C_{S}\right)^{-1} C_{S}^{\top}w}\leq \mu,\nonumber
		\label{ieq:strict_f}
	\end{align}
	then we conclude that $\widehat{z}_{S^c}<1$, which establishes the strict dual feasibility condition.
	On the other hand, one can write
	\begin{equation}
		\begin{aligned}
			4\left\|{C^\top_{S^c}w-C^\top_{S^c}C_{S}\left(C_{S}^{\top} C_{S}\right)^{-1} C_{S}^{\top}w}\right\|_\infty
			&\leq 4\linfinity{C^\top_{S^c}w} + 4\left\|{C^\top_{S^c}C_{S}\left(C_{S}^{\top} C_{S}\right)^{-1} C_{S}^{\top}w}\right\|_\infty\\
			&\leq 4\linfinity{C^\top w} + 4K\gamma\norm{W}_{\max} \linfinity{C^{\top}w}\\
			& = 4(1+K\gamma \norm{W}_{\max})\linfinity{C^\top w}\leq 6\linfinity{C^\top w}
		\end{aligned} \nonumber
	\end{equation}
	where the second inequality follows from Lemma \ref{lem:multual_incoherence}. We also have $C^\top w=C^\top y - C^\top C \Theta_{ij}^\star = \tilde{F}^*_{ij} - \Theta_{ij}^\star - \gamma A^\top G^\top G A\Theta_{ij}^\star$. Therefore, upon choosing 
	\begin{align*}
	   \mu &\geq 6\linfinity{\Theta_{ij}^\star - \tilde{F}^*_{ij}} + 6\gamma K\norm{W}_{\max} D\\
	   &\geq 6\linfinity{\tilde{F}^*_{ij} - \Theta_{ij}^\star - \gamma A^\top G^\top G A\Theta_{ij}^\star},
	\end{align*}
	we establish the strict dual feasibility condition. This implies that $\widehat\Theta_{ij} = (\widehat\Theta_{S;ij}, 0)$ is the unique solution of~\eqref{lasso_problem_2}. Therefore, we have $\supp(\widehat\Theta_{ij})\subseteq\supp(\Theta^\star_{ij})$ and
	\begin{equation}
	\begin{aligned}
		\linfinity{\widehat\Theta_{ij} -\Theta_{ij}^\star}&=\linfinity{\widehat\Theta_{S;ij} - \Theta_{S;ij}^\star} \\
		&\leq \linfinity{(C^\top_SC_S)^{-1}C^\top_S w} + \frac{\mu}{2}\linfinity{(C^\top_SC_S)^{-1}}\\
		&\leq \linfinity{(C^\top_SC_S)^{-1}}\left(\linfinity{C^\top w} + \frac{\mu}{2}\right)\\
		&\leq \frac{ 2\mu}{3+3\gamma\min_i\{\sum_{j\in S^c\backslash i}\abs{W_{i,j}} \}}\leq \frac{2\mu}{3}
	\end{aligned}
	\end{equation}
	Now, it suffices to show that $\supp(\Theta^\star_{ij})\subseteq \supp(\widehat\Theta_{ij})$. Suppose that $\Theta^\star_{k;ij}\not=0$ for some $k$. Then, one can write
	\begin{align*}
	    |\widehat{\Theta}_{k;ij}| &\geq |{\Theta^\star_{k;ij}}| - |\widehat{\Theta}_{k;ij}-{\Theta^\star_{k;ij}}|\\
	    &\geq \Theta_{\min;ij} - \frac{2\mu}{3}>0.
	\end{align*}
Therefore, we have $\supp(\Theta^\star_{ij})\subseteq \supp(\widehat\Theta_{ij})$ which shows that $\supp(\widehat\Theta_{ij})=\supp(\Theta^\star_{ij})$. This completes the proof. $\hfill\square$

\subsection{Proof of Proposition~\ref{prop_sparse_det}}\label{proof_prop_sparse_det}
To prove Proposition~\ref{prop_sparse_det}, we first provide the following lemma adapted from \cite{lee_2013}.
\begin{lem}[Corollary 4.2 of \cite{lee_2013}] \label{lem:l1_consistent}
    Suppose that the irrepresentability assumption holds and 
$$
\|\tilde{F}^*_{ij}\!-\!\Theta^\star_{ij}\|_{\infty}<\frac{\alpha}{8\kappa_{\mathrm{IC}}}\mu.
$$
Then, the optimal solution to \eqref{lasso_problem_1} is unique and satisfies the following properties:
\begin{itemize}
    \item {\it (Estimation error):} $\left\|\widehat\Theta_{ij}-\Theta_{ij}^{\star}\right\|_{2} \leq 2\left(\kappa_{\rho}+\frac{\alpha}{4} \frac{\kappa_{\varrho}}{\kappa_{\mathrm{IC}}}\right) \mu$
    \item {\it (Sparsistency)}: $(B\widehat\Theta_{ij})_{\mathcal{S}_B:^c} = 0 $,
\end{itemize}
where 
\begin{equation}
\begin{aligned}
    \kappa_{\rho}&=\sup _{\theta}\left\{\|B \theta\|_{1} \mid \ltwonorm{\theta}=1, (B\theta)_{\mathcal{S}_B:^c}=0\right\} \\
    \kappa_{\varrho}&=\sup _{\theta}\left\{\|\theta\|_{1} \mid \ltwonorm{\theta}=1, (B\theta)_{\mathcal{S}_B:^c}=0\right\}\\
    \kappa_{\mathrm{IC}} &= \sup\limits_{\norm{z}_{\infty}\leq 1} \norm{B_{\mathcal{S}_B:^c}B_{\mathcal{S}_B:}^\dagger z_{\mathcal{S}_B:}}_{\infty} + \norm{z_{\mathcal{S}_B:^c}}_{\infty}\\
\end{aligned}
\end{equation}
\end{lem}

Based on the above lemma, we proceed with the proof of Proposition~\ref{prop_sparse_det}
\vspace{2mm}

\noindent{\it Proof of Theorem~\ref{prop_sparse_det}.}
Based on our assumptions, Lemma~\ref{lem:l1_consistent} can be invoked to show that $\supp(\widehat\Theta_{ij})\subseteq\supp(\Theta^\star_{ij})$ and $\supp(\widehat\Theta_{k;ij}-\widehat\Theta_{l;ij})\subseteq\supp(\Theta^\star_{k;ij}-\Theta^\star_{l;ij})$. Now, it remains to prove the upper bound on the estimation error, as well as $\supp(\Theta^\star_{ij})\subseteq\supp(\widehat\Theta_{ij})$ and $\supp(\Theta^\star_{k;ij}-\Theta^\star_{l;ij})\subseteq\supp(\widehat\Theta_{k;ij}-\widehat\Theta_{l;ij})$. First, it is easy to verify that 
$\kappa_\rho\leq \sqrt{2\norm{W}_{\max}D_0}+\sqrt{S_0}$ and $\kappa_{\varrho}\leq \sqrt{S_0}$. Moreover, by setting $z_i = 1$ for some $i\in \mathcal{S}^c_{B}$, one can easily verify that $\kappa_{IC}\geq 1$. Therefore, according to Lemma~\ref{lem:l1_consistent}, we have
\begin{align*}
    \left\|\widehat\Theta_{ij}-\Theta_{ij}^{\star}\right\|_{2} &\leq 2\left(\kappa_{\rho}+\frac{\alpha}{4} \frac{\kappa_{\varrho}}{\kappa_{\mathrm{IC}}}\right) \mu\\
    &\leq 2\left(\sqrt{2\norm{W}_{\max}D_0}+\sqrt{S_0}+\frac{\sqrt{S_0}}{4}\right)\mu\\
    &\leq 2\left(\sqrt{2\norm{W}_{\max}D_0}+1.25\sqrt{S_0}\right)\mu,
\end{align*}
where in the second inequality, we used $\kappa_\rho\leq \sqrt{2\norm{W}_{\max}D_0}+\sqrt{S_0}$, $\kappa_{\varrho}\leq \sqrt{S_0}$, $\alpha\leq 1$, and $\kappa_{IC}\geq 1$.

Now, it suffices to show that $\supp(\Theta^\star_{ij})\subseteq\supp(\widehat\Theta_{ij})$ and $\supp(\Theta^\star_{k;ij}-\Theta^\star_{l;ij})\subseteq\supp(\widehat\Theta_{k;ij}-\widehat\Theta_{l;ij})$. Suppose that $\Theta^\star_{k;ij}\not=0$ for some $k$. Then, one can write
\begin{align*}
	    |\widehat{\Theta}_{k;ij}| &\geq |{\Theta^\star_{k;ij}}| \!-\! |\widehat{\Theta}_{k;ij}\!-\!{\Theta^\star_{k;ij}}|\\
	    &\geq \Theta_{\min;ij} \!-\! 2\left(\sqrt{2\norm{W}_{\max}D_0}+1.25\sqrt{S_0}\right)\mu > 0,
	\end{align*}
	where the last inequality follows from the assumption $\mu< \Theta_{\min;ij}/\left(2\left(\sqrt{2\norm{W}_{\max}D_0}+1.25\sqrt{S_0}\right)\right)$. This implies that $\widehat{\Theta}_{k;ij}\not=0$, and hence, $\supp(\Theta^\star_{ij})\subseteq\supp(\widehat\Theta_{ij})$. Similarly, suppose $\Theta^\star_{k;ij} - \Theta^\star_{l;ij}\not=0$ for some $(k,l)$. One can write
	\begin{align*}
	    |\widehat{\Theta}_{k;ij}\!-\!\widehat{\Theta}_{l;ij}| &\!\geq\! |\Theta^\star_{k;ij}\!-\!\Theta^\star_{l;ij}|\!-\!|\widehat\Theta_{k;ij} \!-\! \Theta^\star_{k;ij}| \!-\!|\widehat\Theta_{l;ij} \!-\!\Theta^\star_{l;ij}|\\
	    \geq& \Delta\Theta_{\min; ij} \!-\! 4\left(\sqrt{2\norm{W}_{\max}D_0}\!+\!1.25\sqrt{S_0}\right)\mu\\
	    >& 0,
	\end{align*}
	where the first inequality is due to triangle inequality and the last inequality follows from $\mu< \Delta\Theta_{\min; ij}/\left(4\left(\sqrt{2\norm{W}_{\max}D_0}\!+\!1.25\sqrt{S_0}\right)\right)$. This in turn implies that $\supp(\Theta^\star_{k;ij}-\Theta^\star_{l;ij})\subseteq\supp(\widehat\Theta_{k;ij}-\widehat\Theta_{l;ij})$, which completes the proof.$\hfill\square$

\subsection{Proof of Lemma~\ref{lem:FF}}\label{app_lem:FF}
Using the standard properties of adjucate matrices, one can obtain $F_{vv}$ from $F_{uu}$ after the following steps:
    \begin{enumerate}
        \item Move column $v$ of $M_{uu}$ to position $u$, so that it becomes the $u$-th column of $M_{uu}$. 
        \item Move row $v$ of $M_{uu}$ to position $u$, so that it becomes the $u$-th row of $M_{uu}$.
    \end{enumerate}
    To be more specific, the column and row indices of $M_{uu}$ change from $\{1,2,...,u-1,u+1,...,v-1,v,v+1,...,K\}$ to $\{1,2,...,u-1,v,u+1,...,v-1,v+1,...,K\}$. Since the column/row indices of $M_{vv}$ are $\{1,2,...,u-1,u,u+1,...,v-1,v+1,...,K\}$, we only need to show that $M_{uu}$ and $M_{vv}$ are the same at the $u$-th column and $u$-th row. Moreover, due to symmetry, it suffices to show that $M_{uu}$ and $M_{vv}$ are the same at the $u$-th column. Now, the $u$-th column of $M_{uu}$ is $\left(B_{\mathcal{S}_B:}^\top B_{\mathcal{S}_B:}\right)_{\{1,...,u-1,v,u+1,...,v-1,v+1,...,K\}, v}$ and the $u$-th column of $M_{vv}$ is $\left(B_{\mathcal{S}_B:}^\top B_{\mathcal{S}_B:}\right)_{\{1,...,K\}\backslash v, u}$. From the structure of $B_{\mathcal{S}_B:}^\top B_{\mathcal{S}_B:}$ and the fact that $\Theta_{u;ij}^\star = \Theta_{v;ij}^\star$, we have $\left(B_{\mathcal{S}_B:}^\top B_{\mathcal{S}_B:}\right)_{\{1,...u,...,v,...,K\}, u} = \left(B_{\mathcal{S}_B:}^\top B_{\mathcal{S}_B:}\right)_{\{1,...v,...u,...,K\}, v}$. Thus one can write 
    \begin{equation}
        \begin{aligned}
            {} & \left(B_{\mathcal{S}_B:}^\top B_{\mathcal{S}_B:}\right)_{\{1,...,u-1,v,u+1,...,v-1,v+1,...,K\}, v}
            = \left(B_{\mathcal{S}_B:}^\top B_{\mathcal{S}_B:}\right)_{\{1,...u-1,u,u+1,...,v-1,v+1,...,K\}, u} \\
            &= \left(B_{\mathcal{S}_B:}^\top B_{\mathcal{S}_B:}\right)_{\{1,...,K\}\backslash v, u}
        \end{aligned}
    \end{equation}
    Therefore, $M_{vv}$ can be obtained from $M_{uu}$ using the above two steps.
   This implies that $M_{uu}(-1)^{2|v-u-1|} = M_{vv}$ which in turn leads to $M_{uu} = M_{vv}$. On the other hand, due to the definition of $F_{uu}$ and $F_{vv}$, we have $F_{uu} = (-1)^{2u}M_{uu} = (-1)^{2v}M_{vv} = F_{vv}$. The second statement can be proved in a similar fashion. The details are omitted for brevity.$\hfill\square$

\bibliographystyle{plainnat}
\bibliography{abc}

\end{document}